\documentclass[12pt,english]{article}
\usepackage{mathptmx}
\usepackage[T1]{fontenc}
\usepackage[latin9]{inputenc}
\usepackage{geometry}

\usepackage{tikz}
\usetikzlibrary{patterns}

\usepackage[multiple]{footmisc}
\usepackage{color}
\usepackage{babel}
\usepackage{verbatim}
\usepackage{units}
\usepackage{multirow}
\usepackage{amsmath}
\usepackage{amsthm}
\usepackage{amssymb}
\usepackage{graphicx}
\usepackage{setspace}
\usepackage[authoryear]{natbib}
\usepackage[unicode=true,pdfusetitle,
 bookmarks=true,bookmarksnumbered=false,bookmarksopen=false,
 breaklinks=false,pdfborder={0 0 1},backref=false,colorlinks=true]
 {hyperref}
\hypersetup{
 pdfborderstyle=,citecolor=blue,linkcolor=blue,urlcolor=blue}

\makeatletter

\providecommand{\tabularnewline}{\\}

\theoremstyle{remark}
\newtheorem{rem}{\protect\remarkname}
\theoremstyle{definition}
\newtheorem{defn}{\protect\definitionname}
\theoremstyle{plain}
\newtheorem{thm}{\protect\theoremname}
\theoremstyle{definition}
 \newtheorem{example}{\protect\examplename}
\theoremstyle{plain}
\newtheorem{prop}{\protect\propositionname}
\theoremstyle{plain}
\newtheorem{cor}{\protect\corollaryname}
\theoremstyle{plain}
\newtheorem{lem}{\protect\lemmaname}

\theoremstyle{plain}
\newtheorem{claim}{Claim}

\usepackage{babel}
\usepackage{array}
\usepackage{multirow}
\usepackage{xfrac}

\definecolor{green}{RGB}{0, 150, 0}
\definecolor{blue}{RGB}{0, 0, 200}
\definecolor{red}{RGB}{200, 0, 0}

\makeatother

\providecommand{\corollaryname}{Corollary}
\providecommand{\definitionname}{Definition}
\providecommand{\examplename}{Example}
\providecommand{\lemmaname}{Lemma}
\providecommand{\propositionname}{Proposition}
\providecommand{\remarkname}{Remark}
\providecommand{\theoremname}{Theorem}

\begin{document}
\title{{Communication, Renegotiation and Coordination with Private Values}\thanks{We have benefited greatly from discussions with Srinivas Arigapudi, Tilman Borgers, Michael Greinecker, Jonathan Newton, Bill Sandholm, and Joel Sobel. We would like to express our gratitude to participants of various audiences for many useful comments: LEG2018 \& LEG2019 conferences (in Lund and Bar-Ilan university, respectively), Bielefeld Game Theory 2018 workshop, Israeli Game Theory 2018 conference in IDC, and seminar audiences at Caltech, Tel Aviv University, University of Cyprus, Haifa University, and UC San Diego. Yuval Heller is grateful to the European Research Council for its financial support (\#677057).}}
\author{Yuval Heller\thanks{Department of Economics, Bar Ilan University, Israel. Email: yuval.heller@biu.ac.il.} and Christoph Kuzmics\thanks{Department of Economics, University of Graz, Austria. Email: christoph.kuzmics@uni-graz.at. }}
\maketitle
\begin{abstract}
An equilibrium is communication-proof if it is unaffected by new opportunities to communicate and renegotiate. We characterize the set of equilibria of coordination games with pre-play communication in which players have private preferences over the 
coordinated outcomes. 
{The set of communication-proof equilibria is a small and relatively homogeneous subset of the set of qualitatively diverse Bayesian Nash equilibria.}
Under a communication-proof equilibrium, players never miscoordinate, play their jointly preferred outcome whenever there is one, and communicate only the ordinal part of their preferences. Moreover, such equilibria are robust to changes in players' beliefs and interim Pareto efficient.
\end{abstract}
Final preprint of a manuscript accepted for publication in \emph{Games and Economic Behavior.}

\section{Introduction}
We characterize communication-proof equilibria for a class of coordination games with pre-play cheap-talk communication in which all agents have private information about what action they would prefer to coordinate on. A Bayesian Nash equilibrium is communication-proof if, after the pre-play cheap talk and given the information that this reveals, an opportunity for additional communication 
{cannot lead the players to jointly switch} to a Pareto-improving equilibrium.\footnote{The notion of communication-proofness was introduced by \citet{blume1995communication} in their study of sender-receiver games with one-sided private information.}

We are interested in two typical kinds of situations for which communication-proofness is an appropriate solution concept, albeit
for different reasons in the two situations. The first kind of situation is one in which agents are sophisticated and keep (strategically)
communicating until they reach a mutually beneficial solution. {Communication}-proofness is defined to capture this idea, similarly to the  notions of renegotiation-proofness in contract theory (\citealp{hart1988contract}) or in the repeated games literature (\citealp{farrell1989renegotiation}). As an example, consider a situation of two firms trying to collude by implementing a market-sharing agreement. The agreement is such that each firm is allowed to sell only in specified regions and each firm has private information about which regions they prefer to serve. Another example is two co-authors working on a joint paper and each has to choose whether to write their part of the paper in LaTex or Microsoft Word. Each co-author has private information about their preferred word processor and the intensity of preference, yet both co-authors will gain from coordinating on working with the same word processor. {A similar} situation occurs when firms, when for instance collaborating on research and development, have to choose one of two possible standards. Both firms would benefit from agreeing on the same standard, but each firm has private preferences about which of the options should be the standard. Similarly, {the successful merger of two firms in practice hinges on their ability to efficiently consolidate multiple processes or data sources into a unified system.}

The second kind of situation is one in which communication is feasible and in which behavior is governed by a long-run learning (or evolutionary) process. Communication-proofness here corresponds to a requirement of evolutionary stability at an interim level, when agents can experiment with new behavior that is contingent on the use of additional communication (\citealp{robson1990efficiency}). As an example, consider the problem of two pedestrians suddenly finding themselves face-to-face and trying to get past each other, when they have private information about the direction they want to take after the encounter.\footnote{{A pedestrian can communicate using body gestures. These gestures signal information about the direction, rate, and resoluteness of their proposed course (see, e.g., }\citet[p. 11]{Goffman71}.}$^,$\footnote{{Although the importance of each specific pedestrian encounter is small, the fact that they occur frequently makes their aggregate importance significant, and may justify giving more focus to this class of understudied interactions.}}

The standard solution concept of Bayesian Nash equilibrium is not helpful in predicting whether players can achieve coordination in such incomplete-information settings, how efficient it is if they do, and how communication is used to achieve it. Coordination games with pre-play communication have a wide range of qualitatively very different equilibria.\footnote{{The standard evolutionary refinements are also not very helpful for games with pre-play communication. In a complete information setup, for instance, no strategy satisfies evolutionary stability, while ``too many'' strategies satisfy neutral stability (see \cite{Banerjee00})}.} 
{In addition to fully coordinated} {equilibria, there are babbling equilibria with a high likelihood of miscoordination. Some of these are evolutionarily stable in the absence of communication. There are also equilibria in which agents reveal some information about the intensity of their preferences, and these equilibria often also lead to miscoordination.}

Casual observation, however, suggests that players {typically}
manage to coordinate in such situations. {An illustrative example of this phenomenon can be found in the 1997 series of regional FCC auctions, where licenses for segments of the electromagnetic spectrum were allocated. During this auction, the 'competing' firms used the extremely limited public communication avenues embedded within the trailing digits of their bids to reveal information about their preferred regions. Through this tactic, the firms successfully achieved coordinated collusion, with each one focusing its bids exclusively on its preferred region} (\citealp{cramton2000collusive}).\footnote{In fact, the result of this paper that communication that relies on each player simultaneously sending either 0 or 1 is all that is needed for successful coordination provides another argument against allowing even a brief form of explicit communication between oligopolistic competitors.}

Players also typically coordinate effectively in our {pedestrian} example: Pedestrians typically are able to avoid bumping into each other, even though there is no uniform social norm such as ``always stay on the right'' as there is for cars (\citealp{young1998individual}). Pedestrians often use brief nonverbal communication to signal their preferred direction (e.g., a slight movement to the left or right, a tilt of the head, a glance in a certain direction). The (coordinated) direction in which they pass each other depends on this communication.\footnote{This is motivated by \citet[p. 6]{Goffman71}: ``Take, for example, techniques that pedestrians employ in order to avoid bumping into one another. [...] There are an appreciable number of such devices; they are constantly in use and they cast a pattern on street behavior. Street traffic would be a shambles without them.''}

We show that communication-proof equilibria have a specific structure that is consistent with these casual observations. We show that a strategy is a communication-proof equilibrium if and only if it satisfies the following three independent and easy to verify properties: players never miscoordinate, they play their jointly preferred outcome whenever there is one, and they communicate only {their preferred coordinated outcome without revealing how strongly they prefer this outcome over other coordinated outcomes}. 

The equilibria that satisfy these properties have a simple structure. In all these equilibria communication induces the agents to endogenously face games in which their ordinal preferences are common knowledge. In cases in which agents agree about the optimal joint action, they coordinate efficiently, i.e., on the action that both prefer. In cases where they disagree, they
still coordinate. {However, since the coordinated outcome is not influenced by the strength of each player's preference for one coordinated outcome over another, this coordination typically does not result in ex-ante efficiency.}

We show that communication-proof equilibria do not depend on the distribution of private preferences and are, thus, robust to changes in players' (first- or higher-order) beliefs. In particular, communication-proof equilibrium strategies remain communication-proof even in setups in which the players' distributions of types are interdependent. Also, communication-proof equilibria do not depend on the exact timing of the renegotiation (relative to the communication). We further show that our communication-proof equilibria satisfy appealing efficiency properties  (Section \ref{sec:Efficiency}): {They are Pareto efficient at the interim stage, and they Pareto dominate all equilibria of the game without communication.}

Next, we explore the boundaries of our main result within the class of 2 by 2 coordination games with private information. Our baseline model deals with simple coordination games, where all miscoordinated outcomes {induce the same payoff}. We demonstrate in Section \ref{subsec:Multi-Dimensional-Set}, that there are general coordination games that have 
a communication-proof equilibrium strategy {in which players sometimes miscoordinate}. This makes clear that our main result does not follow from the simple intuition that surely communication must lead to coordination as this is always more efficient. {Nonetheless, we can establish a} 
condition for general coordination games that guarantees that 
communication-proof equilibria are coordinated. {This condition requires that the payoff-dominant
equilibrium of each type coincides with their risk-dominant equilibrium.} 

Finally, we extend our results to more general setups. In Section \ref{sec:Asymmetric-Coordination-Games} we adapt our model and results to asymmetric coordination games. Section \ref{sec:Extreme-Types-with} allows a minority of the players to have non-coordination preferences for which one of the actions is dominant. It shows that an essentially unique strategy that satisfies the above three properties remains communication-proof in this setup. Appendix \ref{subsec:Multiple-Rounds-of} extends our results to multiple rounds of communication.

\paragraph{Relationship to the literature}

Game theorists have long recognized that coordination is an important aspect of successful economic and social interaction, that it requires an explanation even in complete-information coordination games, and that it does not occur in all circumstances. One possible explanation for some, fairly simple, examples of coordination is the concept of a focal point, due to \citet{schelling1960strategy}. This is, loosely speaking, a strategy profile that jumps out at players as clearly the right way to play a game. Perhaps one of the situations in which we most plausibly expect coordination is when people play the same coordination game many times with different people and there is some evolutionary (or learning) process. This approach is already present in the ``mass action'' interpretation of equilibrium given by \citet{nash1950non}. It is then taken up more formally in \citet{smith1973lhe} who define the notion of evolutionary stability. It is well known that all pure equilibria in (complete information) coordination games are evolutionarily stable (whereas mixed equilibria are not stable). This literature thus supports the view that while play in the long run will be coordinated, it is not necessarily efficiently coordinated.\footnote{\citet{kandori1993learning} and \citet{young1993evolution} show that in the long run and under persistent low-probability errors an evolutionary process leads to the risk-dominant, not
necessarily Pareto-dominant, equilibrium.}

Another explanation for coordination is that it is achieved through communication, even if it is simply cheap talk as in \citet{crawford1982strategic}. Early seminal contributions in this direction are \citet{Farrell87} and \citet{rabin1994model}. Communication alone, however, only adds equilibria: the equilibria of the game without communication ``survive'' the introduction of communication as babbling equilibria. The problem, therefore, of how play focuses on the coordinated equilibria does not go away, and one can again appeal to one of the above-mentioned criteria to explain why this might happen.

There is a literature that studies the evolutionary outcome of coordination games with cheap talk, initiated by \citet{robson1990efficiency}. If play is stuck in an inferior equilibrium, a small group of experimenting agents can recognize each other by means of a ``secret handshake.'' They can then play a Pareto-optimal strategy with each other and the inferior equilibrium strategy with agents who are not part of this group, thereby outperforming the agents outside the group.

The above-mentioned literature focuses on complete-information games. However, one of the main reasons why people communicate is that they have privately known preferences that they feel useful to share at least partially before finally choosing actions, as seen in the above examples. One of the main stumbling blocks of studying how communication helps achieve coordination in the presence of incomplete information is that it ``requires overcoming formidable multiple-equilibrium problems'' (\citealp[p. 592]{CH90}). 

We identify \citeauthor{blume1995communication}'s (\citeyear{blume1995communication}) notion of communication-proofness, adapted to our two-sided private information setting, as the appropriate extension of the secret-handshake argument to incomplete-information games. With our characterization result we then show that the plausible refinement of communication-proof equilibria removes this multiplicity problem to a large extent: All communication-proof equilibria, in contrast to Bayesian Nash equilibria, make very similar predictions. {Thus, a key contribution of our paper is showing that the refinement of communication-proofness can be helpful in analyzing games combining incomplete information and pre-play communication.}

\paragraph{Structure}
Section \ref{sec:Model} presents our model. Section \ref{sec:Equilibrium-Strategies} defines Bayesian Nash equilibria and the three key properties that communication-proof equilibria have. Section \ref{sec:Solution-Concept} defines the concept of communication-proofness. Section \ref{sec:Essentially-Unique-Rengotiation-} presents the main result and a sketch of its proof. Section \ref{sec:Efficiency} studies the efficiency properties of communication-proof equilibria. { Next we extend our results to more general setups: multi-dimensional set of types (Section \ref{subsec:Multi-Dimensional-Set}), asymmetric coordination games (Section \ref{sec:Asymmetric-Coordination-Games}), and the presence of a minority of non-coordination types (Section \ref{sec:Extreme-Types-with}). Section \ref{sec:Related-Literature} concludes with a discussion.}
The formal proofs 
are presented in Appendix \ref{sec:Proofs}. {Appendix \ref{app:properties} further studies the relations between the three key properties of Section \ref{sec:Equilibrium-Strategies}. In Appendix \ref{subsec:Multiple-Rounds-of} we extend our analysis to multiple rounds of communication.}

\section{Model\label{sec:Model}}

We consider a setup in which two agents with private idiosyncratic preferences play an (ex-ante) symmetric two-action coordination game that is preceded by pre-play cheap talk.

\paragraph{Players and types}

There are two players, each of which can choose one of two actions, $L$ and $R$. Each player has a privately known ``value'' or ``type.'' The two players' values are independently drawn from a common atomless distribution with a continuous cumulative distribution function $F$ with full support on the unit interval $U=[0,1]$ and with density $f$ (i.e., $f(u)>0$ for each $u\in U$).\footnote{Allowing distributions without full support induces a minor difference in our results: in this setup communication-proofness implies binary communication (as defined in Section \ref{sec:Equilibrium-Strategies}) only of messages that are used with positive probability. With full support it implies binary communication also of unused messages.}

\paragraph{Payoff matrix}

For any realized pair of types, $u$ and $v$, the players play a coordination game given by the payoff matrix given in Table \ref{tab:Payoff-Matrix-of}, where the first entry is the payoff of the player of type $u$ (choosing row) and the second entry is the payoff of the player of type $v$ (choosing column). We call this game the \emph{coordination game without communication} and denote it by $\Gamma$.

\begin{table}[h]
\begin{centering}
\caption{Payoff Matrix of the Coordination Game\label{tab:Payoff-Matrix-of}}
\par\end{centering}
\medskip{}

\centering{} %
\begin{tabular}{cc|cc}
 & \multicolumn{1}{c}{} & \multicolumn{2}{c}{\textcolor{red}{Type }\textcolor{red}{\emph{v}}}\tabularnewline
 &  & \textcolor{red}{\emph{L }}\textcolor{red}{{} }  & \textcolor{red}{\emph{R }}\tabularnewline
\hline
\multirow{2}{*}{\textcolor{blue}{Type }\textcolor{blue}{\emph{u}}} & \textcolor{blue}{\emph{L}}\textcolor{blue}{{} }  & \textcolor{blue}{1-}\textcolor{blue}{\emph{u}},~~\textcolor{red}{1-}\textcolor{red}{\emph{v}}  & \textcolor{blue}{0},~~\textcolor{red}{0} \tabularnewline
 & \textcolor{blue}{\emph{R }}\textcolor{blue}{{} }  & \textcolor{blue}{0},~~\textcolor{red}{0}  & \textcolor{blue}{\emph{u}},~~\textcolor{red}{\emph{v}} \tabularnewline
\end{tabular}
\end{table}

\begin{rem}
{In order to simplify the exposition of our model and to ease its notation, our baseline model focuses on symmetric coordination games with a one-dimensional set of coordination types. We extend our results to more general setups in later sections: general coordination games with a multi-dimensional set of types (Section \ref{subsec:Multi-Dimensional-Set}), asymmetric coordination games (Section \ref{sec:Asymmetric-Coordination-Games}), and allowing a minority of non-coordination (or dominant action) types (Section \ref{sec:Extreme-Types-with}).}
\end{rem}
\paragraph{Interpretation of the model and the motivating examples}

In the example of two firms trying to collude by market-sharing, choosing the same action corresponds to dividing the market such that each firm is a monopolist in one of the two regions. Choosing different actions corresponds to the firms competing in the same region, which yields a low profit normalized to zero. A firm's type $u$ corresponds to how profitable it is for the firm to be a monopolist in one region relative to being one in the other region. Similarly, in the second example of co-authors coordinating on the word processor, the players get a low payoff (normalized to 0) if they use different word processors. The difference between a player's type $u$ and $1-u$ corresponds to how much the author's ease of use in one word processor is larger than in the other one.

In the final motivating example of pedestrians suddenly finding themselves face to face and trying to get past each other, each action corresponds to the direction in which the pedestrians turn to avoid bumping into each other. When both pedestrians choose the same side (say, each pedestrian chooses her left), the pedestrians do not bump into each other. When they choose different sides they do bump into each other, in which case they get a low payoff normalized to zero. A pedestrian's type reflects her private preference for the direction in which she would like to turn to avoid a collision due to the direction she plans to take after the encounter. That is, a type $u>\sfrac12$ corresponds to a pedestrian who plans to head right after the encounter. For such a type choosing \emph{R} is more convenient than choosing \emph{L} as it induces a shorter walking path.

\paragraph{Pre-play communication}

After learning their type, but before playing this coordination game, the two players each simultaneously send a publicly observable message from a finite set of messages $M$ (satisfying $4\leq\left|M\right|<\infty$). We denote by $\Delta(M)$ the set of all probability distributions over messages in $M$.\footnote{Our results essentially remain the same if $M$ is countably infinite. The assumption of $\left|M\right| \geq 4$ implies that a single round of communication during the renegotiation stage can achieve a sufficient degree of communication for our main results to hold (see Section \ref{sec-CP}). Our results remain the same for $M=2$ if one allows the players during the renegotiation stage to either have two stages of communication or to rely on a (binary) sunspot.}$^,$\footnote{In Appendix \ref{subsec:Multiple-Rounds-of} we show that, communication-proof equilibria in coordination games are unaffected by the length (number of rounds) of communication (in contrast to the results in other setups of incomplete-information games; see, e.g., \citealp{aumann2003long}). {Moreover, our results  remain the same even when transitioning from simultaneous to sequential communication, as long as players are allowed to observe a sunspot, which represents a public realization of a random variable. This is because the joint lotteries required to implement most communication-proof equilibria (with the exception of $\sigma_L$ and $\sigma_R$) necessitate either simultaneous communication or the presence of a sunspot.}} We assume that messages are costless. We call the game, so amended, the \emph{coordination game with communication} and denote it by $\langle\Gamma,M\rangle$.

\paragraph{Strategies}

A player's (ex-ante) strategy in the coordination game with communication is then a pair $\sigma=(\mu,\xi)$. The (Lebesgue measurable) \emph{message function} $\mu:U\to\Delta(M)$ describes which (possibly random) message is sent for each possible realization of the agent's type.  The \emph{action function} $\xi:M\times M\rightarrow U$ describes the maximal type (cutoff type) that chooses $L$ as a function of the observed message profile. That is, when an agent who follows strategy $(\mu,\xi)$ observes a message profile $\left(m,m'\right)$ (message $m$ sent by the agent, and message $m'$ sent by the opponent), then the agent plays $L$ if her type $u$ is at most $\xi\left(m,m'\right)$ (i.e., if $u\le\xi(m,m')$), and she plays $R$ if $u>\xi(m,m')$.\footnote{In Appendix \ref{subsec:Any-Best-Reply-Strategy-is-cutoff} we show that the restriction to cut-off strategies is without loss of generality: any ``generalized'' strategy $\xi:M\times M\rightarrow U$ is dominated by a strategy with a ``threshold'' action function.} (The choice that the threshold type plays $L$ does not affect our analysis, given the assumption of $F$ being atomless.) Let $\Sigma$ be the set of all strategies in the game $\langle\Gamma,M\rangle$.

Let $\mu_{u}\left(m\right)$ denote the probability, given message function $\mu$, that a player sends message $m$ if she is of type $u$. Let $\bar{\mu}\left(m\right)=\mathbb{E}_{u}\left[\mu_{u}\left(m\right)\right]$ be the mean probability that a player of a random type sends message $m$ (where the expectation is taken with respect to $F$). Let $\mbox{supp}\left(\bar{\mu}\right)=\left\{ m\in M\mid\bar{\mu}\left(m\right)>0\right\}$ denote the support of $\bar{\mu}$. We say that message $m$ is in the support of $\sigma=\left(\mu,\xi\right)$, denoted by $m \in \mbox{supp}(\sigma)$, if $m\in\mbox{supp}\left(\bar{\mu}\right)$.

With a slight abuse of notation we write $\xi\left(m,m'\right)=L$ when all types (who send message $m$ with positive probability) play $L$ (i.e., when $\xi\left(m,m'\right)\geq\sup\left(u\in U|\mu_{u}\left(m\right)>0\right)$), and we write $\xi\left(m,m'\right)=R$ when all types play $R$ (i.e., when $\xi\left(m,m'\right)\leq\inf\left(u\in U|\mu_{u}\left(m\right)>0\right)$).

\section{Equilibrium Strategies\label{sec:Equilibrium-Strategies} and Three Key Properties}

We here define the standard notion of (symmetric Bayesian Nash) equilibrium strategies, present the three key properties that communication-proof equilibria turn out to have, and present examples of equilibria in the coordination game with communication with and without these properties. These equilibria are illustrated in Figure \ref{fig:equilibria} at the end of this section.

Given a strategy profile $\left(\sigma,\sigma'\right)$ and a type profile $u,v\in U$, let $\pi_{u,v}\left(\sigma,\sigma'\right)$ denote the payoff of a player of type $u$ who follows strategy $\sigma$ and faces an opponent of type $v$ who follows strategy $\sigma'$. Formally, for $\sigma=(\mu,\xi)$ and $\sigma'=(\mu',\xi')$,
\begin{eqnarray*}
\pi_{u,v}\left(\sigma,\sigma'\right)=\sum_{m\in M}\sum_{m'\in M}\mu_{u}\left(m\right)\mu_{v}\left(m'\right)\left((1-u)\boldsymbol{1}_{\{u\le\xi(m,m')\}}\boldsymbol{1}_{\{v\le\xi'(m',m)\}}\right.+\left.u\boldsymbol{1}_{\{u>\xi(m,m')\}}\boldsymbol{1}_{\{v>\xi'(m',m)\}}\right),
\end{eqnarray*}
where $\boldsymbol{1}_{\{x\}}$ is the indicator function equal to $1$ if statement $x$ is true and zero otherwise. Let
\[
\pi_{u}\left(\sigma,\sigma'\right)=\mathbb{E}_{v}\left[\pi_{u,v}\left(\sigma,\sigma'\right)\right]\equiv\int_{v=0}^{1}\pi_{u,v}\left(\sigma,\sigma'\right)f\left(v\right)dv
\]
denote the expected interim payoff of a player of type $u$ who follows strategy $\sigma$ and faces an opponent with a random type who follows strategy $\sigma'$. Finally, let,
\[
\pi\left(\sigma,\sigma'\right)=\mathbb{E}_{u}\left[\pi_{u}\left(\sigma,\sigma'\right)\right]\equiv\int_{u=0}^{1}\pi_{u}\left(\sigma,\sigma'\right)f\left(u\right)du
\]
denote the ex-ante expected payoff of an agent who uses strategy $\sigma$ against strategy $\sigma'$.

A strategy $\sigma$ is a \emph{(symmetric Bayesian Nash) equilibrium strategy} if $\pi_{u}\left(\sigma,\sigma\right)\geq\pi_{u}\left(\sigma',\sigma\right)$ for each $u\in U$ and each strategy $\sigma'\in\Sigma$. Let $\mathcal{E}\subseteq\Sigma$ denote the set of all equilibrium strategies of $\langle\Gamma,M\rangle$.

\paragraph{Three key properties}

We call a strategy $\sigma=(\mu,\xi)\in\Sigma$ \emph{mutual-preference consistent} if whenever $u,v<\sfrac{1}{2}$ then $\xi\left(m,m'\right)=\xi\left(m',m\right)=L$ for all $m\in\mbox{supp}(\mu_{u})$ and all $m'\in\mbox{supp}(\mu_{v})$ and if whenever $u,v>\sfrac{1}{2}$ then $\xi\left(m,m'\right)=\xi\left(m',m\right)=R$ for all $m\in\mbox{supp}(\mu_{u})$ and all $m'\in\mbox{supp}(\mu_{v})$. That is, players with the same ordinal preference coordinate on their mutually preferred outcome.

We call a strategy \emph{coordinated} if $\xi\left(m,m'\right)=\xi\left(m',m\right)\in\left\{ L,R\right\}$ for any messages $m,m'\in\mbox{supp}(\bar{\mu})$. A coordinated strategy leads to a coordinated outcome with probability one.

For any message $m\in M$, define the expected probability of a player's opponent playing $L$ conditional on the player sending $m$ and the opponent following 
$\sigma=(\mu,\xi)\in\Sigma$, as
\[
\beta^{\sigma}(m)=\int_{u=0}^{1}\sum_{m'\in\mbox{supp}(\mu_{u})}\mu_{u}(m')\boldsymbol{1}_{\left\{ u\leq\xi(m',m)\right\} }f(u)du.
\]
We say that strategy $\sigma$ has \emph{binary communication} if there are two numbers $0 \le \underline{\beta}^{\sigma} \le \overline{\beta}^{\sigma} \le 1$ such that for all messages $m\in M$ we have $\beta^{\sigma}(m)\in[\underline{\beta}^{\sigma},\overline{\beta}^{\sigma}]$, for all messages $m\in M$ such that there is a type $u<\sfrac{1}{2}$ with $\mu_{u}(m)>0$ we have $\beta^{\sigma}(m)=\overline{\beta}^{\sigma}$, and for all messages $m\in M$ such that there is a type $u>\sfrac{1}{2}$ with $\mu_{u}(m)>0$ we have $\beta^{\sigma}(m)=\underline{\beta}^{\sigma}$. That is, binary communication implies that players (essentially) use just two kinds of messages: any message sent by types $u<\sfrac12$ induces the same consequence of maximizing the probability of the opponent playing $L$, and any message sent by types $u>\sfrac12$ induces the opposite consequence of maximizing the opponent's probability of playing $R$. Note that, as defined here, a strategy, in which the player's message does not affect the probability of the partner playing $L$, has binary communication.

In Appendix \ref{app:properties} we show that none of these three properties is implied by the other two.

\paragraph{Left tendency $\alpha^{\sigma}$} Consider a strategy that satisfies the above three properties. Coordination and mutual-preference consistency jointly determine the behavior of agents with the same ordinal preferences (i.e., when both types are below $\sfrac{1}{2}$, or both above $\sfrac{1}{2})$. Binary communication, then, implies the following: The probability with which the players coordinate on $L$, conditional on having different ordinal preferences (i.e., conditional on one player having type $u<\sfrac{1}{2}$ and the other player having type $v>\sfrac{1}{2}$), is independent on the message sent by the player. We denote this probability by $\alpha^{\sigma}$, and refer to it as the \emph{left tendency} of the strategy. We can express $\underline{\beta}^{\sigma}$ and $\overline{\beta}^{\sigma}$ as follows:
\[
\underline{\beta}^{\sigma}=F(\sfrac{1}{2})\alpha^{\sigma}\,\,\,\,\,\,\,\mbox{and}\,\,\,\,\,\,\,\overline{\beta}^{\sigma}=F(\sfrac{1}{2})+\left(1-F(\sfrac{1}{2})\right)\alpha^{\sigma}.
\]
The first equality ($\underline{\beta}^{\sigma}=F(\sfrac{1}{2})\alpha^{\sigma}$) is implied by the fact that when any type $u>\sfrac{1}{2}$ sends a message expressing her preference for coordination on $R$, the players coordinate on $L$ only if the opponent's preferred outcome is $L$ (which happens with a probability of $F(\sfrac{1}{2})$). They then coordinate on $L$ with a probability of $\alpha^{\sigma}$. The second equality ($\overline{\beta}^{\sigma}=F(\sfrac{1}{2})+\left(1-F(\sfrac{1}{2})\right)\alpha^{\sigma}$) follows from the following observation: When any type $u<\sfrac{1}{2}$ sends a message expressing her preference for coordination on $L$, the players coordinate on $L$ with probability one if the opponent's preferred outcome is $L$, and they coordinate on $L$ with a probability of $\alpha^{\sigma}$ if the opponent's preferred action is $R$.

\paragraph{Examples of equilibria satisfying all properties}

The following strategies, denoted by $\sigma_{L}$, $\sigma_{R}$, and $\sigma_{C}$, are prime examples (that play a special role in later sections) of strategies that are all mutual-preference consistent and coordinated and have binary communication.

The strategies $\sigma_{L}$ and $\sigma_{R}$ are given by the pairs $\left(\mu^{*},\xi_{L}\right)$ and $\left(\mu^{*},\xi_{R}\right)$, respectively. The message function $\mu^{*}$ has the property that there are messages $m_{L},m_{R}\in M$ such that message $m_{L}$ indicates a preference for $L$ and $m_{R}$ a preference for $R$, and the action functions $\xi_{L}$ and $\xi_{R}$ are defined as follows:
\[
\mu^{*}\left(u\right)=\begin{cases}
m_{L} & u\leq\sfrac{1}{2}\\
m_{R} & u>\sfrac{1}{2}.
\end{cases}\,\,\,\,\,\,\,\,\xi_{L}\left(m,m'\right)=\begin{cases}
R & m=m'=m_{R}\\
L & \mbox{otherwise},
\end{cases}\,\,\,\,\,\,\,\,\xi_{R}\left(m,m'\right)=\begin{cases}
L & m=m'=m_{L}\\
R & \mbox{otherwise}.
\end{cases}
\]
This means that the ``fallback norm'' of $\sigma_{L}$ (which is applied when the agents have different preferred outcomes) is to coordinate on $L$, while that of $\sigma_{R}$ is to coordinate on $R$. In other words the left tendency of $\sigma_{L}$ is one and the left tendency of $\sigma_{R}$ is zero.

Strategy $\sigma_{C}=\left(\mu_{C},\xi_{C}\right)$ has the ``fallback norm'' of using a joint lottery to choose the coordinated outcome. Each agent simultaneously sends a random bit and the coordinated outcome depends on whether the random bits are equal or not.

We denote four distinct messages by $m_{L,0},m_{L,1},m_{R,0},m_{R,1}\in M$, where we interpret the first subscript ($R$ or $L$) as the agent's preferred direction, and the second subscript ($0$ or $1$) as a random binary number chosen with probability $\sfrac{1}{2}$ each by the agent. Formally, the message function $\mu_{C}$ is defined as follows:
\[
\mu_{C}\left(u\right)=\begin{cases}
\sfrac{1}{2}m_{L,0}\oplus\sfrac{1}{2}m_{L,1} & u\leq\sfrac{1}{2}\\
\sfrac{1}{2}m_{R,0}\oplus\sfrac{1}{2}m_{R,1} & u>\sfrac{1}{2},
\end{cases}
\]
where $\alpha m\oplus(1-\alpha)m'$ is a lottery with a probability of $\alpha$ on message $m$ and $1-\alpha$ on message $m'$. In the second stage, if both agents share the same preferred outcome they play it. Otherwise, they coordinate on $L$ if their random numbers differ, and coordinate on $R$ otherwise. Formally:
\[
\xi_{C}\left(m,m'\right)=\begin{cases}
R & \left(m,m'\right)\in\left\{ \left(m_{R,0},m_{R,0}\right),\left(m_{R,0},m_{R,1}\right),\left(m_{R,0},m_{L,0}\right),\left(m_{R,1},m_{L,1}\right)\right.\\
 & \,\,\,\,\,\,\,\,\,\,\,\,\,\,\:\,\,\,\,\,\,\,\,\,\,\,\,\,\,\left.\left(m_{R,1},m_{R,1}\right),\left(m_{R,1},m_{R,0}\right),\left(m_{L,0},m_{R,0}\right),\text{\ensuremath{\left(m_{L,1},m_{R,1}\right)}}\right\} \\
L & \mbox{otherwise}.
\end{cases}
\]

The outcome of $\sigma_{C}$ can also be implemented by a fair joint lottery that determines which of the two players determines the coordinated action used by both players. This alternative implementation yields exactly the same outcome: if both agents share the same preferred outcome they play it, and conditional on the agents disagreeing on the preferred outcome, they coordinate on each action with equal probability.

\paragraph{One-dimensional set of strategies satisfying the properties}

The set of strategies with the above three properties (coordination, mutual-preference consistency, and binary communication) is essentially one-dimensional because the left tendency $\alpha^{\sigma}\in[0,1]$ of such a strategy $\sigma$ describes all payoff-relevant aspects. Two strategies with the same left tendency can only differ in the way in which the players implement the joint lottery when they have different preferred outcomes. These implementation differences are nonessential, as the probability of the joint lottery inducing the players to coordinate on $L$ remains the same.

Any left-tendency $\alpha^{\sigma}\in[0,1]\cap \mathbb{Q}$ can be implemented by a jointly controlled lottery in which the players send random messages in such a way that they are indifferent between all messages, and the joint distribution of messages induces $\alpha^{\sigma}$ (\citealp{maschler1968repeated}). This is demonstrated for $\alpha^{\sigma}=\sfrac12$ in the strategy $\sigma^{C}$ presented above.

Note that of all the strategies that satisfy the three properties, strategies $\sigma_{L}$ and $\sigma_{R}$ are the \emph{simplest} in terms of the number of ``bits'' needed to implement the message function. Strategy $\sigma_{C}$ is in a certain sense the \emph{fairest}: conditional on a coordination conflict, i.e., conditional on one agent having a type between $0$ and $\sfrac{1}{2}$ and the other agent having a type between $\sfrac{1}{2}$ and $1$, both agents expect the same payoff. By contrast, strategy $\sigma_{L}$ favors types below $\sfrac{1}{2}$, and strategy $\sigma_{R}$ favors types above $\sfrac{1}{2}$.

\paragraph{Examples of equilibria not satisfying some of the properties}

The coordination game with communication $\langle\Gamma,M\rangle$ admits many more equilibria that satisfy only some or even none of the three properties defined above.

First, the game admits babbling equilibria, which do not satisfy mutual-preference consistency. Each babbling equilibrium can be identified with an $x\in[0,1]$ that satisfies $F(x)=x$, where agents choose $L$ if and only if  their type is below $x$. The case of $x=1$ (resp., $x=0$) corresponds to a uniform norm of always playing $L$ (resp., $R$). A case of $x \in (0,1)$ corresponds to an inefficient babbling equilibria, in which agents sometimes miscoordinate.

The game also admits equilibria in which agents reveal some information about the intensity of their preferences (i.e., some information beyond only stating whether $u\leq\sfrac{1}{2}$ or $u>\sfrac{1}{2}$). One simple example of such an equilibrium is Example \ref{exa:high-ex-ante-miscoordination} in Section \ref{sec:Efficiency}. 

\paragraph{Illustration of equilibria and the first-best outcome}

Figure \ref{fig:Illustration-of-Equilibrium} illustrates five of the equilibria described above: the equilibria that satisfy the three key properties: $\sigma_{L}$, $\sigma_{R}$, and $\sigma_{C}$, the babbling equilibrium of always playing $R$, and the equilibrium $\sigma_{\mbox{ex}}$, which satisfies none of the three key properties. It also depicts (in the bottom right panel) the first-best outcome in which the players reveal their types and then coordinate on the action that maximizes the sum of payoffs (i.e., the players coordinate on $L$ if $u+v\leq1$ and they coordinate on $R$ if $u+v>1$). This is not an equilibrium: each player has an incentive to present a more extreme type than her real type (e.g., all types $u>\sfrac{1}{2}$ would claim to have type $1$).

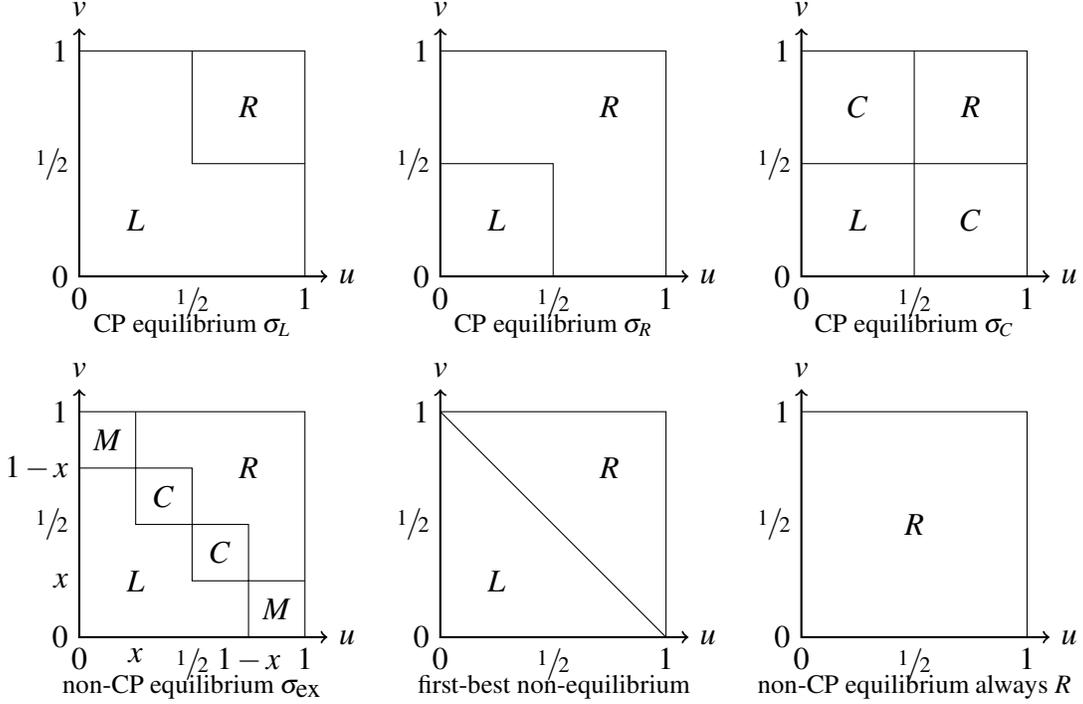
\begin{figure}[h]
\begin{center}
\begin{tikzpicture}[scale=0.6]

\draw [<->,thick] (0,5.5) node (yaxis) [above] {$v$}
        |- (5.5,0) node (xaxis) [right] {$u$};
\draw [-] (0,5) node[left] {$1$} -- (5,5) -- (5,0) node[below] {$1$};
\draw (0,0) node[below] {$0$} node[left] {$0$};
\draw (2.5,0) node[below] {$\sfrac12$};
\draw (1.25,0) node[below] {$x$};
\draw (3.75,0) node[below] {$1-x$};
\draw (0,2.5) node[left] {$\sfrac12$};
\draw (0,1.25) node[left] {$x$};
\draw (0,3.75) node[left] {$1-x$};
\draw (1.25,5) -- (1.25,3.75) -- (2.5,3.75) -- (2.5,2.5) -- (3.75,2.5) -- (3.75,1.25) -- (5,1.25);
\draw (0,3.75) -- (1.25,3.75) -- (1.25,2.5) -- (2.5,2.5) -- (2.5,1.25) -- (3.75,1.25) -- (3.75,0);
\draw (1.25,1.25) node {$L$};
\draw (3.75,3.75) node {$R$};
\draw (1.875,3.125) node {$C$};
\draw (3.125,1.875) node {$C$};
\draw (0.625,4.375) node {$M$};
\draw (4.375,0.625) node {$M$};
\draw (2.5,-1.1) node {\footnotesize{non-CP equilibrium $\sigma_{\mbox{ex}}$}};

\draw [<->,thick] (8,5.5) node (yaxis) [above] {$v$}
        |- (13.5,0) node (xaxis) [right] {$u$};
\draw [-] (8,5) node[left] {$1$} -- (13,5) -- (13,0) node[below] {$1$};
\draw (8,0) node[below] {$0$} node[left] {$0$};
\draw (10.5,0) node[below] {$\sfrac12$};
\draw (8,2.5) node[left] {$\sfrac12$};
\draw (8,5) -- (13,0);
\draw (9.25,1.25) node {$L$};
\draw (11.75,3.75) node {$R$};
\draw (10.5,-1.1) node {\footnotesize{first-best non-equilibrium}};

\draw [<->,thick] (16,13.5) node (yaxis) [above] {$v$}
        |- (21.5,8) node (xaxis) [right] {$u$};
\draw [-] (16,13) node[left] {$1$} -- (21,13) -- (21,8) node[below] {$1$};
\draw (16,8) node[below] {$0$} node[left] {$0$};
\draw (18.5,8) node[below] {$\sfrac12$};
\draw (16,10.5) node[left] {$\sfrac12$};
\draw (18.5,8) -- (18.5,13);
\draw (16,10.5) -- (21,10.5);
\draw (17.25,9.25) node {$L$};
\draw (19.75,11.75) node {$R$};
\draw (17.25,11.75) node {$C$};
\draw (19.75,9.25) node {$C$};
\draw (18.5,6.9) node {\footnotesize{CP equilibrium $\sigma_{C}$}};

\draw [<->,thick] (16,5.5) node (yaxis) [above] {$v$}
        |- (21.5,0) node (xaxis) [right] {$u$};
\draw [-] (16,5) node[left] {$1$} -- (21,5) -- (21,0) node[below] {$1$};
\draw (16,0) node[below] {$0$} node[left] {$0$};
\draw (18.5,0) node[below] {$\sfrac12$};
\draw (16,2.5) node[left] {$\sfrac12$};
\draw (18.5,2.5) node {$R$};
\draw (18.5,-1.1) node {\footnotesize{non-CP equilibrium always $R$}};

\draw [<->,thick] (0,13.5) node (yaxis) [above] {$v$}
        |- (5.5,8) node (xaxis) [right] {$u$};
\draw [-] (0,13) node[left] {$1$} -- (5,13) -- (5,8) node[below] {$1$};
\draw (0,8) node[below] {$0$} node[left] {$0$};
\draw (2.5,8) node[below] {$\sfrac12$};
\draw (0,10.5) node[left] {$\sfrac12$};
\draw (2.5,13) -- (2.5,10.5) -- (5,10.5);
\draw (1.25,9.25) node {$L$};
\draw (3.75,11.75) node {$R$};
\draw (2.5,6.9) node {\footnotesize{CP equilibrium $\sigma_{L}$}};

\draw [<->,thick] (8,13.5) node (yaxis) [above] {$v$}
        |- (13.5,8) node (xaxis) [right] {$u$};
\draw [-] (8,13) node[left] {$1$} -- (13,13) -- (13,8) node[below] {$1$};
\draw (8,8) node[below] {$0$} node[left] {$0$};
\draw (10.5,8) node[below] {$\sfrac12$};
\draw (8,10.5) node[left] {$\sfrac12$};
\draw (8,10.5) -- (10.5,10.5) -- (10.5,8);
\draw (9.25,9.25) node {$L$};
\draw (11.75,11.75) node {$R$};
\draw (10.5,6.9) node {\footnotesize{CP equilibrium $\sigma_{R}$}};

\end{tikzpicture}
\end{center}

\caption{\label{fig:Illustration-of-Equilibrium}{\textbf{Six example strategies.}\label{fig:equilibria} The axis represent the two players' types $u$ and $v$. Letters $L$, $R$, $C$, and $M$ represent coordination on $L$, coordination on $R$, coordination on a random action, and miscoordination (both players playing their preferred action), respectively.}}
\end{figure}

\section{Definition of Communication-Proofness\label{sec:Solution-Concept}}\label{sec-CP}

For any given strategy in $\Sigma$ employed by both players in the game $\langle\Gamma,M\rangle$, communication and knowledge of this strategy lead to updated and possibly, different and asymmetric information about the two agents' types. Suppose that the updated distributions of types are given by some distribution functions $G$ and $H$. The two agents then face a (possibly asymmetric) game of coordination without communication, which we shall denote by $\Gamma(G,H)$. Note that the original game (without communication) $\Gamma$ is then given by $\Gamma(F,F)$.

Let $f_{m}$ be the type density conditional on the agent following a given strategy in the game $\langle\Gamma,M\rangle$ and sending a message\footnote{The density $f_{m}$ depends on the given strategy in the game $\langle\Gamma,M\rangle$. For aesthetic reasons we refrain from giving this strategy a name and from indicating this obvious dependence in our notation.} $m\in\mbox{supp}(\bar{\mu})$, i.e., $f_{m}(u)=\sfrac{f(u)\mu_{u}(m)}{\bar{\mu}(m)}$, and let $F_{m}$ be the cumulative distribution function associated with density $f_{m}$.

We allow players to renegotiate after communication. Renegotiating players can use a new round of communication. Given a strategy of the game $\langle\Gamma,M\rangle$ employed by both players, we denote the induced renegotiation game after a positive probability message pair $m,m'\in M$ by $\langle\Gamma(F_{m},F_{m'}),M\rangle$.
Let $\pi_{u}\left(\sigma,\sigma'|H\right)$ denote the expected payoff for the player using strategy $\sigma$ of type $u$ given strategy profile $(\sigma,\sigma')$ in game $\langle\Gamma(G,H),M\rangle$:
\[
\pi_{u}\left(\sigma,\sigma'|H\right)=\mathbb{E}_{v\sim H}\left[\pi_{u,v}\left(\sigma,\sigma'\right)\right] \equiv \int_{v=0}^{1}\pi_{u,v}\left(\sigma,\sigma'\right)h\left(v\right)dv,
\]
and similarly let $\pi_{v}\left(\sigma,\sigma'|G\right)$ denote the expected payoff for the player using strategy $\sigma'$ of type $v$ in  $\langle\Gamma(G,H),M\rangle$:
\[
\pi_{v}\left(\sigma,\sigma'|G\right)=\mathbb{E}_{u\sim G}\left[\pi_{u,v}\left(\sigma,\sigma'\right)\right] \equiv \int_{u=0}^{1}\pi_{u,v}\left(\sigma,\sigma'\right)g\left(u\right)du.
\]
Let $\mathcal{E}(G,H)$ be the set of all (possibly asymmetric) equilibria of the coordination game with communication $\langle\Gamma(G,H),M\rangle$. Let $\pi_{u}^{(m,m')}(\sigma,\sigma'|H_{m'})$ (resp., $\pi_{v}^{(m,m')}(\sigma,\sigma'|G_m)$) denote the \emph{post-communication payoff} for the player using strategy $\sigma$ (resp., $\sigma'$) of a type $u$ (resp., $v$) given $\left(\sigma=(\mu,\xi),\sigma'=(\mu',\xi')\right)$ in game $\langle\Gamma(G,H),M\rangle$ conditional on message pair $m \in \mbox{supp}(\sigma),m' \in \mbox{supp}(\sigma'):$
\[
\pi_{u}^{(m,m')}(\sigma,\sigma'|H) = \left\{ \begin{array}{cc} (1-u) H_{m'}\left(\xi'(m',m)\right) & \mbox{ if } u\le\xi(m,m') \\ u\left(1-H_{m'}\left(\xi'(m',m)\right)\right) & \mbox{ if } u>\xi(m,m') \\ \end{array} \right.
\]
\[
\pi_{v}^{(m,m')}(\sigma,\sigma'|G) = \left\{ \begin{array}{cc} (1-v)G_m\left(\xi(m,m')\right) & \mbox{ if } v\le\xi'(m',m) \\ v\left(1-G_m\left(\xi(m,m')\right)\right) & \mbox{ if } v>\xi'(m',m). \\ \end{array} \right.
\]

Following \cite{blume1995communication}, we say that a strategy profile $\left(\tau,\tau'\right)$ CP trumps another profile $\left(\sigma,\sigma'\right)$ if there is a possible pair of messages such that, given the information induced by the message pair, the profile $\left(\tau,\tau'\right)$, using another round of communication, yields a Pareto-improvement over the post-communication expected payoffs induced by $\left(\sigma,\sigma'\right)$. 
\begin{defn}
Strategy profile $\left(\tau,\tau'\right) \in \Sigma^2$ \emph{CP trumps} strategy profile $\left(\sigma,\sigma'\right) \in \Sigma^2$ with respect to distribution profile $(G,H)$ and message profile $m\in \mbox{supp}\left(\sigma\right),m'\in \mbox{supp}\left(\sigma'\right)$ if\footnote{{We do not model the alternative strategy profile $\tau,\tau'$ as a strategy profile in a game that includes two rounds of communication. Following \citet{blume1995communication}, we rather model $\tau,\tau'$ as a strategy profile in a communication game with a single round of communication (which is the renegotiation round), where the information about the opponent's type from the first round is embedded through changing the distributions of types from $(G,H)$ to the conditional distributions $(G_m,H_m')$.}}$^,$
\footnote{For conceptual consistency we could additionally require that a CP-trumping strategy profile be symmetric after a pair of identical messages. We refrain from imposing this, as it would make the notation cumbersome and would not change the set of (strongly or weakly) communication-proof strategies in our setting.}
\begin{enumerate}
\item $\left(\tau,\tau'\right)\in\mathcal{E}(G_{m},H_{m'}), \text{ and}$
\item $\pi_{u}(\tau,\tau'|H_{m'}) \ge \pi_{u}^{(m,m')}(\sigma,\sigma'|H)$ and $\pi_{v}(\tau,\tau'|G_{m}) \ge \pi_{v}^{(m,m')}(\sigma,\sigma'|G),$ for all $u\in\mbox{supp}(G_{m})$ and all $v\in\mbox{supp}(H_{m'})$ with strict inequality for some $u\in\mbox{supp}(G_{m})$ or some $v\in\mbox{supp}(H_{m'})$.
\end{enumerate}
\end{defn}

We say that a strategy $\sigma$ is strongly communication-proof if for any possible message profile, there does not exist a new equilibrium, which might require another round of communication, that Pareto-dominates the post-communication payoff of $\sigma$. The weaker notion of weak communication-proofness allows such a Pareto-improving equilibrium to exist as long as this latter equilibrium is not stable in the sense that it is CP trumped by another equilibrium. Formally:
\begin{defn}
An equilibrium strategy $\sigma\in\mathcal{E}$ is \emph{strongly communication-proof} if it is not CP trumped by any strategy profile with respect to $\left(F,F\right)$ 
{and any message profile.}
\end{defn}

\begin{defn}
An equilibrium strategy $\sigma\in\mathcal{E}$ is \emph{weakly communication-proof} if for any strategy profile $\left(\tau,\tau'\right)$ that CP trumps $\left(\sigma,\sigma\right)$ with respect to $\left(F,F\right)$ and message profile $\left(m,m'\right)$, there exists a strategy profile $\left(\rho,\rho'\right)$ that CP trumps $\left(\tau,\tau'\right)$ with respect to $\left(F_{m},F_{m'}\right)$.
\end{defn}
Observe that in games with complete information, our two notions coincide. Moreover, they are both equivalent to the Pareto frontier of the set of Nash equilibria, i.e., to the subset of Nash equilibria that are not Pareto-dominated by other Nash equilibria.
\citet{blume1995communication} present a related notion of communication-proofness defined in the spirit of \citeauthor{morgenstern1953theory}'s \citeyearpar{morgenstern1953theory} set stability:\footnote{\citet{blume1995communication} use this notion for sender-receiver games, in which there is incomplete information only on one side, but it is straightforward to adapt to their notion to games with incomplete information on both sides.} the set of equilibria is divided into stable and unstable equilibria; a strategy profile is communication-proof $\grave{\textrm{a}}$ la \citeauthor{blume1995communication} if it is not CP trumped by a stable equilibrium; and the set of stable equilibria 
(which is not necessarily unique)
is defined consistently (any stable equilibrium is only CP trumped by unstable equilibria, and any unstable equilibrium is CP trumped by some stable equilibrium).  \citeauthor{blume1995communication} show that any sender-receiver game (in which only one player has private information and her set of actions is a singleton) admits a communication-proof equilibrium.

Observe that \citeauthor{blume1995communication}'s notion is in between our two notions. The fact that a strongly communication-proof equilibrium 
is not CP trumped by any strategy profile implies that any strongly communication-proof equilibrium is communication-proof $\grave{\textrm{a}}$ la  \citeauthor{blume1995communication}. The fact that any equilibrium that CP trumps a communication-proof equilibrium is unstable (and, thus, CP-trumped by another equilibrium) implies that any communication-proof equilibrium $\grave{\textrm{a}}$ la  \citeauthor{blume1995communication} is weakly communication-proof.

\subsection{Evolutionary Interpretation of Communication-Proofness}
Consider a setup in which agents in a large population are repeatedly and randomly matched to play a game. Assume that the agents' behavior is influenced by evolutionary forces, such that the share of agents who play actions leading to higher payoffs gradually increases. \cite{robson1990efficiency} {considers}  complete information games in which players can freely communicate before playing the game. \citeauthor{robson1990efficiency} 
suggests that Pareto-efficiency within the set of Nash equilibria is a plausible refinement  for capturing stable behavior in these setups. Suppose that equilibrium $\sigma$ is Pareto inferior to another equilibrium $\sigma'$. A population state in which everyone plays the inferior equilibrium $\sigma$ is not stable in the following sense: A small group of experimenting agents (''mutants``) can recognize each other by means of a ``secret handshake.'' When they do, they play a Pareto-improving equilibrium $\sigma'$ with each other. Otherwise they play the inferior equilibrium $\sigma$ with agents who are not part of this group, thereby outperforming the agents outside the group.\footnote{\label{footnote-evol-papers}There are various ways to assess the foundations of the secret handshake argument in complete information coordination games {with pre-play communication}, see e.g., \cite{matsui1991cheap, Warneryd_1991, blume1993evolutionary, warneryd1993cheap, sobel1993evolutionary, kim1995evolutionary, Banerjee00, santos2011co}. 
{As argued by} \cite{Schlag_1993_cheapWP,schlag1994does}, {if the set of messages is finite, the secret handshake argument may fail, allowing for the persistence of Pareto-inefficient equilibria. This is because experimenting agents might not have unused messages to identify each other. The secret handshake argument can be formally established, ensuring only efficient equilibria can be stable, in setups with infinitely many messages} (\citealp{bhaskar1998noisy}) {or in setups with a finite message set and voluntary communication}
(\citealp{hurkens2003evolutionary}).
We conjecture that similar foundations for the secret handshake argument can also be provided for incomplete information games.}

We study situations in which players have private types, and can freely communicate after they learn their own types. We argue that the natural way to extend the above evolutionary-motivated refinement of Pareto efficiency to the current setup is by the refinement of communication-proofness. The argument is similar to that of\cite{robson1990efficiency}: If there exists a pair of messages $m,m'$ after which the players play an equilibrium $\sigma$ that is Pareto inferior for all possible types relative to another equilibrium $\sigma'$ (where the latter equilibrium $\sigma'$ might require an additional round of communication to implement), then a small group of experimenting agents can recognize each other by a ``secret handshake'', and play the Pareto-improving equilibrium $\sigma'$ after observing messages $m,m'$. For example, a state in which all pedestrians choose $L$ is unstable to the presence of a few experimenting agents who will touch their right ear (or slightly nod to the right, or who use whatever other communication ``device'' that \cite[p. 6]{Goffman71} discusses) if their preferred action is R, and will play R if both players have touched their right ears (or use whatever other ``device'').

\section{Main Result}\label{sec-main}

\label{sec:Essentially-Unique-Rengotiation-} Our main result shows that both of our notions of communication-proofness coincide in our setup, and they are characterized by satisfying the three key properties of Section \ref{sec:Equilibrium-Strategies}.
\begin{thm}
\label{thm:uniqueness}Let $\sigma$ be a strategy of the game with communication $\langle\Gamma,M\rangle$. The following three statements are equivalent:
\begin{enumerate}
\item $\sigma$ is mutual-preference consistent, coordinated, and has binary communication.
\item $\sigma$ is a strongly communication-proof equilibrium strategy.
\item $\sigma$ is a weakly communication-proof equilibrium strategy.
\end{enumerate}
\end{thm}
\begin{proof}[Sketch of proof]
The proof that ``1'' implies ``2'' is fairly straightforward (and is proven in Appendix \ref{subsec:uniqueness}). The proof implies, in particular, that $\sigma_{L}$, $\sigma_{R}$,
and $\sigma_{C}$ are not CP trumped by any other strategy profiles. It is immediate that ``2'' implies ``3.'' We here provide a sketch of the proof that ``3'' implies ``1.'' The proof in Appendix
\ref{subsec:uniqueness} is split into three lemmas, each showing that one of the three properties must hold.

Lemma \ref{lemma:nomiscoordination} proves that a weakly communication-proof equilibrium strategy must be coordinated: if play after any message pair is not coordinated then it is CP trumped in
the renegotiation game by either $\sigma_{L}$, $\sigma_{R}$, or $\sigma_{C}$. To see this, suppose first that both players use thresholds below $\sfrac{1}{2}$. Then this strategy is Pareto-dominated by $\sigma_{R}$ as types above $\sfrac{1}{2}$ gain because $\sigma_{R}$ induces their first-best outcome, and types below $\sfrac{1}{2}$ gain because $\sigma_{R}$ yields a higher coordination probability and a higher probability of the opponent playing this type's preferred action $L$. Analogously, an equilibrium in which both players use thresholds above $\sfrac{1}{2}$ is Pareto-dominated by $\sigma_{L}$. Suppose, finally, that player one uses threshold $x<\sfrac{1}{2}$, while player two uses threshold $x'>\sfrac{1}{2}$. Observe that $x<\sfrac{1}{2}$ (resp., $x'>\sfrac{1}{2}$) can be an equilibrium threshold only if player two (resp., player one) plays $L$ with an average probability of less (resp., more) than $\sfrac{1}{2}$. This, implies that players
in these equilibria coordinate with a probability of at most $\sfrac{1}{2}$, and one can show that such a low coordination probability implies that these equilibria are Pareto-dominated by $\sigma_{C}$.

Next, we show in Lemma \ref{lemma:binary} that a weakly communication-proof equilibrium strategy must have binary communication. The reason for this is that if a strategy is coordinated, then different messages can only lead to different ex-ante probabilities of coordination on $L$ (and $R$). Thus, any type who favors $L$, i.e., any type $u<\sfrac{1}{2}$, will choose a message to maximize this probability, while any type $u>\sfrac{1}{2}$ will choose a message to minimize this probability. Thus, essentially only two kinds of messages are used in a coordinated equilibrium strategy.

Finally, we show in Lemma \ref{lemma:mpc} that a weakly communication-proof equilibrium strategy must be mutual-preference consistent. Given that it is coordinated, we know that any message pair will lead to either coordination on $L$ or on $R$. If it is not mutual-preference consistent then, without loss of generality, there are two types $u,u'<\sfrac{1}{2}$ that, with positive probability, send a message pair $(m,m')$ that leads them to coordinate on $R$. But then all types who send this
message pair would be weakly better off (and some strictly better off) if instead of coordinating on $R$ they use strategy $\sigma_{R}$, which would allow them to coordinate on $L$ if and only if both types are below $\sfrac{1}{2}$.
\end{proof}
As the two notions of communication-proofness coincide in our setup, we henceforth omit the word ``weakly''/ ``strongly'' and write communication-proof equilibrium strategy to describe either of our (equivalent) solution concepts. Note that the set of communication-proof equilibria is completely independent of the distribution $F$ (i.e., for any two distributions of types $F$ and $F'$, strategy $\sigma$ is a communication-proof equilibrium in  $\Gamma(F)$ if and only if it is communication-proof  $\Gamma(F')$.) It is not difficult to show that this implies that any communication-proof equilibrium strategy remains communication-proof even in setups in which the distributions of types are correlated, and in setups in which different types have different beliefs about the opponent's type.

\section{On Efficiency \label{sec:Efficiency}}

In this section we explore the efficiency properties of communication-proof equilibria. {We begin with a negative result by demonstrating that a non-communication-proof equilibrium  can induce a higher ex-ante payoff than all communication-proof equilibria.}

In the subsequent parts of Section \ref{sec:Efficiency} {we show favorable properties of communication-proof equilibria. First, we show that each communication-proof equilibrium is Pareto efficient at the interim stage (after each player knows her own type). Next, we show two additional appealing properties when focusing on $\sigma_L$, $\sigma_R$ and $\sigma_M$: either $\sigma_L$ or $\sigma_R$  provides the highest ex-ante payoff of all the coordinated equilibria, and any equilibrium without communication is strictly Pareto-dominated by either $\sigma_L$, $\sigma_R,$ or $\sigma_M$.}

\paragraph{{Negative result: not maximizing the ex-ante payoff}}

{Equilibria featuring miscoordination can incentivize agents to genuinely disclose cardinal information regarding their type.} This can happen if there is a message that induces a higher probability of coordinating on the agent's preferred outcome but also a higher probability of miscoordination compared with some other available message. Such a message can then be chosen by extreme types with $u$ far from $\sfrac{1}{2}$, while moderate types with $u$ closer to $\sfrac{1}{2}$ choose the other message. Such equilibria with miscoordination may induce a higher ex-ante payoff {than all communication-proof equilibria}, if the benefit from signaling the extremeness of the type outweighs the loss due to miscoordination. Consider the following example.
\begin{example}
\label{exa:high-ex-ante-miscoordination} For simplicity we let the distribution of types $F$ be discrete with four atoms $\sfrac{1}{10}+\epsilon$, $\sfrac{1}{2}-\epsilon$, $\sfrac{1}{2}+\epsilon$, $\sfrac{9}{10}-\epsilon$, with a probability of $\sfrac{1}{4}$ for each atom and $\epsilon>0$
sufficiently small.\footnote{One can easily adapt the example to an atomless distribution of types,
in which each atom is replaced with a continuum of nearby types.} The game admits three babbling equilibria: always coordinating on $L$, always coordinating on $R$, both with an ex-ante payoff of $\sfrac{1}{2}$, and playing $L$ iff the type is less than $\sfrac{1}{2}$ with an ex-ante payoff of $\sfrac{7}{20}<\sfrac{1}{2}$ for all $\epsilon$ sufficiently small. Theorem \ref{thm:uniqueness} (together with the symmetry of the distribution $F$) implies that with communication, any communication-proof equilibrium strategy (in particular $\sigma_{L}$ or $\sigma_{R}$) induces the same expected ex-ante payoff of $\sfrac{3}{5}>\sfrac{1}{2}$ for all $\epsilon$ sufficiently small.

This game also has a (non-communication-proof) equilibrium strategy with miscoordination that yields a higher ex-ante payoff than the communication-proof payoff of $\sfrac{3}{5}$, provided that the message set $M$ has sufficiently many elements. To simplify the presentation we here allow the players to use public correlation devices to determine their joint play after sending messages, which can be approximately implemented by a sufficiently large message set (\`a la \citealp{maschler1968repeated}).
Let $m_{L},m_{l},m_{r},m_{R}\in M$ and consider strategy $\sigma=(\mu,\xi)$ as follows. Let $\mu(\sfrac{1}{10}+\epsilon)=m_{L}$, $\mu(\sfrac{1}{2}-\epsilon)=m_{l}$, $\mu(\sfrac{1}{2}+\epsilon)=m_{r}$, and $\mu(\sfrac{9}{10}-\epsilon)=m_{R}$, and let $\xi(m_{a},m_{b})=L$ if $a,b\in\{L,l\}$, $\xi(m_{a},m_{b})=R$ if $a,b\in\{r,R\}$, $\xi(m_{L},m_{r})=\xi(m_{r},m_{L})=L$, $\xi(m_{l},m_{R})=\xi(m_{R},m_{l})=R$, $\xi(m_{l},m_{r})=\xi(m_{r},m_{l})$ be a joint lottery to coordinate on $L$ or $R$ with probability $\sfrac{1}{2}$ each, and, finally, let $\xi(m_{L},m_{R})=\xi(m_{R},m_{L})$ be a joint lottery to coordinate on $L$ or $R$ with probability $\sfrac{3}{10}$ each, and to play the inefficient mixed equilibrium (in which each type plays her preferred outcome with probability $\sfrac{9}{10}-\epsilon$) with probability $\sfrac{4}{10}$. It is straightforward to verify that for, say $\epsilon=\sfrac{1}{100}$, this strategy is indeed an equilibrium strategy with an ex-ante payoff of around $0.627$, which is higher than the ex-ante payoff of $\sfrac{3}{5}$ of all the communication-proof equilibria. This equilibrium strategy is not coordinated (nor does it satisfy the other two properties of mutual-preference consistency and binary communication) and hence, by Theorem \ref{thm:uniqueness}, it is not communication-proof.
\end{example}

\paragraph{Interim Pareto Efficiency}
Example \ref{exa:high-ex-ante-miscoordination} {demonstrates that communication-proof equilibria might not be Pareto efficient at the ex-ante stage (before the players know their types). In contrast, we now show that any communication-proof equilibrium is Pareto efficient at the interim stage: after each player knows her type, and before communicating with the other player. In order to formally analyze interim efficiency, we present an auxiliary definition of a social choice function.}

An (ex-ante symmetric) \emph{social choice function} is a function $\phi:\left[0,1\right]^{2}\rightarrow\Delta\left(\left\{ L,R\right\}^{2}\right)$ assigning to each pair of types a possibly correlated profile with the condition that $\phi_{u,v}(a,b)=\phi_{v,u}(b,a)$ for any $a,b\in\{L,R\}$, where\footnote{We restrict attention to symmetric social choice functions in order to maintain our focus on symmetric equilibria, and in order to allow us to use a simpler notation without player subscripts. Proposition \ref{prop:Pareto-optimal} below, however, also holds even if we allow asymmetric social choice functions.} $\phi_{u,v}\equiv\phi\left(u,v\right)$. We interpret $\phi_{u,v}$ as the correlated action profile played by the two players when a player of type $u$ interacts with a player of type $v$. Let $\Phi$ be the set of all such functions.
Any strategy of any coordination game with communication induces a social choice function in $\Phi$, but not all social choice functions in $\Phi$ can be generated by a strategy of a given coordination game with communication. One can interpret $\Phi$ as the set of outcomes that can be implemented by a designer who perfectly observes the types of both players and, can force the players to play arbitrarily.

For each type $u\in\left[0,1\right]$, let $\pi_{u}\left(\phi\right)$ denote the expected payoff of a player of type $u$ under social choice function $\phi$, i.e., $\pi_{u}\left(\phi\right)=\mathbb{E}_{v}\left[\left(1-u\right)\phi_{u,v}\left(L,L\right)+u\phi_{u,v}\left(R,R\right)\right].$

A strategy is interim Pareto-dominated if there is a social choice function that is weakly better for all types, and strictly better for some types.
\begin{defn}
A strategy $\sigma\in\Sigma$ is \emph{interim Pareto-dominated} by function $\phi\in\Phi$ if $\pi_{u}\left(\sigma,\sigma\right)\leq\pi_{u}\left(\phi\right)$ for each type $u\in[0,1]$, with a strict inequality for a positive measure set of types.
\end{defn}
Note that our {definition of interim Pareto domination is permissive} in the sense that we allow
the designer to perfectly observe the players' types, and to enforce non-Nash play on the players. A strategy $\sigma\in\Sigma$ is \emph{interim Pareto {efficient}} if it is not interim Pareto-dominated by any $\phi\in\Phi$. It is immediate that any interim Pareto efficient equilibrium strategy is communication proof (because any CP trumping strategy profile is interim Pareto dominant). Our next result shows that the converse is true in our setup, namely that all communication-proof equilibria satisfy our strong requirement of interim Pareto efficiency. That is, even a designer with perfect ability to observe the players' types and to enforce any behavior cannot achieve a Pareto improvement with respect to any communication-proof equilibrium strategy.\footnote{As discussed in the extended working paper version, \citet{heller2021}, the result that any communication-proof equilibrium is interim Pareto-efficient holds also for asymmetric equilibria. Moreover, two of these asymmetric communication-proof equilibria are also ex-ante Pareto efficient: the equilibrium that always chooses the action preferred by Player 1, and the analogous equilibrium that always chooses the action preferred by Player 2.}
\begin{prop}
\label{prop:Pareto-optimal}Every communication-proof equilibrium strategy is interim Pareto {efficient}.
\end{prop}
\begin{proof}[Sketch of proof; see Appendix \ref{app:proofsefficiency} for the formal proof.]
Recall that by Theorem \ref{thm:uniqueness} and the discussion about the one-dimensional set of strategies in Section \ref{sec:Equilibrium-Strategies}, any communication-proof equilibrium strategy $\sigma$ is characterized by its left tendency $\alpha^{\sigma}$. In order for a social choice function $\phi$ to improve the payoff of any type $u<\sfrac{1}{2}$ (resp., $u>\sfrac{1}{2}$) relative to the payoff induced by $\sigma$, it must be that $\phi$ induces any $u<\sfrac{1}{2}$ (resp., $u>\sfrac{1}{2}$) to coordinate on $L$ with probability larger (resp., smaller) than $\alpha^{\sigma}$. This implies that the probability of two players coordinating on $L$, conditional on the players having different preferred outcomes, must be larger (resp., smaller) than $\alpha^{\sigma}$. However, these two requirements contradict each other.
\end{proof}
\paragraph{{Either $\sigma_L$ or $\sigma_R$ maximize ex-ante payoff}}
{Next we show an appealing property of the two simplest communication-proof equilibrium strategies. Specifically, we show that the ex-ante expected payoff of either $\sigma_L$ or $\sigma_R$ is higher than the ex-ante payoff of any coordinated equilibrium.}

\begin{prop}
\label{prop-maximizing-coordinated} Let $\sigma\in\mathcal{E}$ be a coordinated equilibrium strategy. Then
\[
\pi\left(\sigma,\sigma\right)\leq\max\left\{ \pi\left(\sigma_{L},\sigma_{L}\right),\pi\left(\sigma_{R},\sigma_{R}\right)\right\} .
\]
\end{prop}
\begin{proof}[Sketch of proof; see Appendix \ref{app:proofsefficiency} for the formal proof.] Let $\alpha^{\sigma}$ be the probability of two players who each follow $\sigma$ to coordinate on $L$, conditional on the players having different preferred outcomes. It is easy to see that $\sigma$ is dominated by the communication-proof equilibrium strategy with the same left tendency $\alpha^{\sigma}$, and that the payoff of the latter strategy is a convex combination of the payoffs of $\sigma_{L}$ and $\sigma_{R}$, which implies that $\pi\left(\sigma,\sigma\right)\leq\max\left\{ \pi\left(\sigma_{L},\sigma_{L}\right),\pi\left(\sigma_{R},\sigma_{R}\right)\right\}$.
\end{proof}
\begin{rem}
\label{rem-all-stage}
One could refine the notion of communication-proofness to allow agents to renegotiate to a Pareto-improving equilibrium also at earlier stages (\`a la \citealp{benoit1993renegotiation}): at the interim stage before observing the realized messages induced by the original equilibrium, and at the ex-ante stage before each agent observes her own type. Proposition \ref{prop:Pareto-optimal} implies that allowing agents to renegotiate also at the interim stage does not change the set of communication-proof equilibria. Proposition \ref{prop-maximizing-coordinated} implies that if $\pi\left(\sigma_{L},\sigma_{L}\right)\neq\pi\left(\sigma_{R},\sigma_{R}\right)$ then allowing agents to renegotiate also at the ex-ante stage yields a unique ``all-stage'' communication-proof equilibrium. This equilibrium is either $\sigma_{L}$ or $\sigma_{R}$. The set of communication-proof equilibria is not affected by introducing ex-ante renegotiation if $\pi\left(\sigma_{L},\sigma_{L}\right)=\pi\left(\sigma_{R},\sigma_{R}\right)$. We do not allow ex-ante communication in our model because in most relevant applications (and, in particular, in the motivating examples presented earlier), it seems plausible that the agents can only communicate after they know their own types.
\end{rem}
\paragraph{{Any babbling equilibrium is dominated by either $\sigma_L$, $\sigma_R$ or $\sigma_C$}}
Recall that any babbling equilibrium (or any equilibrium in the coordination game without communication) is characterized by a cutoff value $x\in[0,1]$ such that $x=F(x)$ with the interpretation that types $u\le x$ play $L$ and types $u>x$ play $R$. Let $\pi_{u}\left(x,x'\right)$ denote the payoff of an agent with type $u$ who follows a strategy with cutoff $x$ and faces a partner of unknown type who follows a strategy with cutoff $x'$:
\[
\pi_{u}\left(x,x'\right)=\boldsymbol{1}_{\{u\leq x\}}F\left(x'\right)\left(1-u\right)+\boldsymbol{1}_{\{u>x\}}\left(1-F\left(x'\right)\right)u,
\]
and let $\pi\left(x,x'\right)=\mathbb{E}_{u}\left[\pi_{u}\left(x,x'\right)\right]$ be the ex-ante expected payoff of an agent who follows $x$ and faces a partner who follows $x'$. {Our final result of this section shows an appealing property of the three simplest communication-proof equilibrium strategies. Specifically,} it shows that any (possibly asymmetric) {babbling } equilibrium  is strictly Pareto-dominated by either $\sigma_{L}$, $\sigma_{R}$, or $\sigma_{C}$.
\begin{cor}
\label{cor:strictly-improving-no-communication} Let $\left(x,x'\right)$
be a (possibly asymmetric) {babbling} equilibrium. Then $\pi_{u}\left(x,x'\right)\leq\pi_{u}\left(\sigma_{L},\sigma_{L}\right)$ for all types $u\in U$, or $\pi_{u}\left(x,x'\right)\leq\pi_{u}\left(\sigma_{R},\sigma_{R}\right)$ for all types $u\in U$, or $\pi_{u}\left(x,x'\right)\leq\pi_{u}\left(\sigma_{C},\sigma_{C}\right)$ for all types $u\in U$. Moreover, all the inequalities are strict for almost all types.
\end{cor}
Corollary \ref{cor:strictly-improving-no-communication} is immediately implied by Lemma \ref{lemma:nomiscoordination} in Appendix \ref{subsec:uniqueness}, and the sketch of proof of the lemma is presented as part of the sketch of the proof of Theorem \ref{thm:uniqueness}.

\section{Multidimensional Sets of Types\label{subsec:Multi-Dimensional-Set}} 

In our model we made the simplifying assumption that miscoordination provides the same payoff (normalized to zero) to both players. This is not completely innocuous. In this section we explore which results are still true in this more general setting. Consider the following multidimensional set of types. Let $\hat{U}$, a subset of $\mathbb{R}^{4}$, be the set of payoff matrices of binary coordination games, with $u_{ab}$ being the payoff if a player chooses action $a\in\{L,R\}$ while her opponent chooses action $b\in\{L,R\}$:
\[
\hat{U}=\left\{ \left(u_{LL},u_{LR},u_{RL},u_{RR}\right)\mid u_{LL}>u_{RL}\mbox{ and }u_{RR}>u_{LR}\right\} .
\]

{Thus, for all types the best reply against each opponent's action is to play the same action.} Let $\hat{\Gamma}=\hat{\Gamma}\left(G\right)$ denote the coordination game with the type space $\hat{U}$, endowed with an atomless CDF $G$ over $\hat{U}$ with a density $g$. Similarly, let $\langle\hat{\Gamma},M\rangle$ be the corresponding game with communication.

Given a type $u=\left(u_{LL},u_{LR},u_{RL},u_{RR}\right)$, let $\varphi_{u}\in\left[0,1\right]$ denote type $u$'s \emph{indifference threshold}, which is the probability of the opponent playing $L$ that induces an agent of type $u$ to be indifferent:
$\varphi_{u}=\sfrac{(u_{RR}-u_{LR})}{(u_{LL}-u_{RL}+u_{RR}-u_{LR})}.$

Observe that an agent with indifference threshold $\varphi_{u}$, where $\varphi_{u}$ is a number always between $0$ and $1$, prefers to play $L$ ($R$) if her partner plays $L$ with probability larger (smaller) than $\varphi_{u}$. In other words, for a given probability of her partner playing $L$, a type $u$ prefers to play $L$ if and only if $\varphi_{u}$ is less than that probability. Thus, the indifference threshold $\varphi_{u}$ replaces what we denoted by $u$ in the main model. In particular, in this setting we can also restrict attention to cutoff action functions. These are now applied to $\varphi_{u}$ instead of to $u$. Thus, under a strategy $\sigma=(\mu,\xi)$ a player plays action $L$ after observing a message pair $(m,m')$ if and only if $\varphi_{u}\le\xi(m,m')$.
We set
\[
F(\varphi)=\int_{\left\{ u\in U:\varphi_{u}\leq\varphi\right\}} g(u)du
\]
to be the implied distribution over the players' indifference threshold induced by density $g$. As in the baseline model, we assume that $F(\varphi)$ has full support on the interval $[0,1]$.

Recall that a pure equilibrium in a two-action game is \emph{risk-dominant} (\citealp{harsanyi1988general}) if for each type $u$, playing her part of the equilibrium is a best reply against the opponent randomizing equally over the two actions (which holds iff $\varphi_{u}\leq\sfrac{1}{2}\,\Leftrightarrow\,u_{LL}-u_{LR}\ge u_{RR}-u_{RL}.$). The crucial assumption that we implicitly make in our main model is that the payoff-dominant equilibrium of each type coincides with her risk-dominant equilibrium. That is, the coordination preferences are \emph{uniform} in the sense that the two notions of dominance agree with each other for all types.
\begin{defn}
An atomless distribution over the space $\hat U$ with density function $g:\hat{U} \to \mathbb{R}$ satisfies \emph{{uniform} coordination preferences} if for any $u\in \hat{U}$ with $g(u)>0$ we have $u_{LL}\ge u_{RR}\,\,\Leftrightarrow\,\,\varphi_{u}\leq\sfrac{1}{2}$.
\end{defn}
In other words, we assume that all types prefer coordinating on action $L$ iff action $L$ is also their best-reply against a uniform distribution of actions of the opponent. 
{In particular, in the pedestrians motivating example, this assumption implies that} if a pedestrian's preferred outcome is coordination on $L$, she would also choose $L$ when facing a pedestrian who chooses each side with equal probabilities).

It is immediate that one direction of Theorem \ref{thm:uniqueness} (namely, ``$1\Rightarrow 2\Rightarrow3$'') {holds} in this multi-dimensional setup without any additional assumptions (specifically, all the relevant arguments in the proof of Theorem \ref{thm:uniqueness} hold in this more general setup). Our next result shows that the {proof of the} other direction of Theorem \ref{thm:uniqueness} (namely, ``$3\Rightarrow 1$'') goes through unchanged if we assume that the all types have uniform coordination preferences.
\begin{thm}[Theorem \ref{thm:uniqueness} adapted to a multidimensional set of types]
\label{Thm:multidimensional}Let $\sigma$ be a strategy in a game $\langle\hat{\Gamma},M\rangle$ satisfying {uniform} coordination preferences. Then the following statements are equivalent:
\begin{enumerate}
\item $\sigma$ is mutual-preference consistent, coordinated, and has binary communication.
\item $\sigma$ is a strongly communication-proof equilibrium strategy.
\item $\sigma$ is a weakly communication-proof equilibrium strategy.
\end{enumerate}
\end{thm}
The proof is presented in Appendix \ref{subsec:Proof-of-Theorem-multi}. The intuition is the same as in Theorem \ref{thm:uniqueness}. The following two examples demonstrate why the restriction of uniform coordination preferences is necessary for the ``$3\Rightarrow1$'' part of the result. {Specifically, the examples show how preferences that do not satisfy uniform coordination allows equilibria with miscoordination to be communication proof.} Both examples are presented with discrete distributions, but it is straightforward to modify them to nearby full-support atomless distributions.

\begin{example}[\emph{Strongly communication-proof equilibrium with miscoordination}]
\label{example:simplecounterexample}
There are two possible preference types with equal probabilities as follows:
\begin{center}
\begin{tabular}{c|cc}
$u_{L}$ & L & R\tabularnewline
\hline 
L & 10 & 0\tabularnewline
R & 9 & 5\tabularnewline
\end{tabular}~~~%
\begin{tabular}{c|cc}
$u_{R}$ & L & R\tabularnewline
\hline 
L & 5 & 9\tabularnewline
R & 0 & 10\tabularnewline
\end{tabular}
\par\end{center}
Let $M=\{m_{L},m_{R}\}$ and let $\sigma=(\mu,\xi)$ be a strategy with the following properties. Each player reveals her preferred outcome, $\mu\left(u_{L}\right)=m_{L}$, $\mu\left(u_{R}\right)=m_{R}$. If the players sent the same messages they both coordinate on their jointly preferred outcome. If the  players sent different messages (which implies that one player has type $u_L$ and the opponent has type $u_R$) then they play the mixed equilibrium of the induced complete-information game. That is, player $u_{L}$ plays $L$ with probability $\sfrac{1}{6}$ and player $u_{R}$ plays $R$ with probability $\sfrac{1}{6}$. This yields an expected payoff of $8$ plus $\sfrac{1}{3}$ to each player. Observe that strategy $\sigma$ is an equilibrium, in which each type of player at the interim expects a payoff of $\sfrac{1}{2} \cdot 10 + \sfrac{1}{2} \left(8+\sfrac{1}{3}\right) = 9 + \sfrac{1}{6}$. Misreporting one's type would lead to a best possible payoff of $\sfrac{1}{2} \sfrac{17}{3} + \sfrac{1}{2} \cdot 5$, which is less than the equilibrium payoff.  

Observe that $\sigma$ is strongly communication-proof. Both players achieve their maximal feasible payoff if they send the same message. If they send different messages, then the equilibrium payoff of $8 + \sfrac{1}{3}$ to both players is Pareto-efficient in the convex hull of the set of Nash equilibria. It, thus, cannot be CP-trumped (recall that in complete information games communication can only implement outcomes in the convex hull of the set of Nash equilibria). 
\end{example}
Example \ref{example:simplecounterexample} could be extended to many distributions over types. For instance, consider a distribution over types that attaches positive weight (which might be close to one) to a (suitably chosen - see below) subset $U'$ of the set of types in the basic model (as in Table \ref{tab:Payoff-Matrix-of}). The remaining positive weight (which might be arbitrarily small) is attached to a pair of preferences $u,v \in \hat{U}$ that satisfy $u_{LL} > u_{RR}$ and $v_{LL} < v_{RR}$. Suppose further that the complete information coordination game between $u$ and $v$ has a Pareto-undominated mixed equilibrium. This implies that there is a convex combination of coordinated equilibria $(L,L)$ and $(R,R)$, with $\alpha \in [0,1]$ the weight on $(L,L)$, that is Pareto-dominated by the Pareto-undominated mixed equilibrium, for the complete information coordination game between $u$ and $v$. Now, consider the strategy $\sigma$ that is such that all types indicate their ordinal coordination preferences (by sending messages $m_L$ or $m_R$) plus the types $u$ and $v$ fully reveal their preferences (by sending distinct messages $m_u$ and $m_v$, respectively). Suppose further that play after $(m_L,m_L)$ and after $(m_L,m_u)$ is $(L,L)$ and after $(m_R,m_R)$, after $(m_R,m_v)$ it is $(R,R)$, and after $(m_L,m_R)$, $(m_u,m_R)$, and $(m_L,m_v)$ it is coordinated with left-tendency $\alpha$ and, finally, after $(m_u,m_v)$ it is the Pareto-undominated mixed equilibrium. Then $\sigma$ is strongly renegotiation-proof. The subset $U'$ of types can be chosen such that no type prefers to send a different message from the one they are supposed to, play after any message pair is in equilibrium, and there is no other equilibrium that CP-trumps $\sigma$.\footnote{If, for instance, the additional preference pair $(u,v)=(u_L,u_R)$ as in Example \ref{example:counterexample}, we could choose $\alpha=\sfrac{1}{2}$ and $U'=\left[\sfrac{2}{5},\sfrac{3}{5}\right]$.}

The non-uniform coordination types in Example\ref{example:simplecounterexample} {are ``extreme'' in the sense that they prefer a miscoordinated outcome over a coordinated outcome (e.g., type $u_L$ prefers miscoordination on $(R,L)$ over coordination on $(R,R)$). This allows the mixed equilibrium with miscoordination in the complete-information game played between two different types to be Pareto-undominated, which, in turn, implies communication-proofness. Our next example, shows that the direction of ``3 $\Rightarrow$ 1'' does not hold even with (non-uniform coordination) types for which any {miscoordinated} outcome is Pareto dominated by any coordinated outcome.}
\begin{example}
\label{example:counterexample}
There are four possible preference types with probabilities given below:

\begin{center}
\begin{tabular}{cccc}
\begin{tabular}{c|cc}
$u_{L_{1}}$  & L  & R \\
\hline
L  & 2  & 0 \\
R  & 0  & 1 \\
\end{tabular}
&
\begin{tabular}{c|cc}
$u_{L_{2}}$  & L  & R \\
\hline
L  & 2  & -15 \\
R  & 0  & 1 \\
\end{tabular}
&
\begin{tabular}{c|cc}
$u_{R_{1}}$  & L  & R \\
\hline
L  & 1  & 0 \\
R  & 0  & 2 \\
\end{tabular}
&
\begin{tabular}{c|cc}
$u_{R_{2}}$  & L  & R \\
\hline
L  & 1  & 0 \\
R  & -15  & 2 \\
\end{tabular} \\ 
$P(u_{L_1})=\sfrac{1}{18}$ & $P(u_{L_2})=\sfrac{8}{18}$ & $P(u_{R_1})=\sfrac{1}{18}$ & $P(u_{R_2}) = \sfrac{8}{18}$ 

\end{tabular}
\end{center}

The letter of each type ($L$ or $R$) represents its preferred coordinated outcomes. The less frequent ``1'' types have standard uniform coordination preferences as in the baseline model. The more frequent ``2'' types have non-uniform preferences: although type $u_{L_2}$ prefers coordination on $(L,L)$, its risk-dominant equilibrium is $(R,R)$ because playing $L$ is risky due to the low payoff of $-15$ obtained from outcome $(L,R)$.  
Let $M=\{m_{L},m_{R}\}$ and let $\sigma=(\mu,\xi)$ be such that each player reveals her preferred coordinated outcome (i.e., $\mu\left(u_{L_{1}}\right)=\mu\left(u_{L_{2}}\right)=m_{L}$ and $\mu\left(u_{R_{1}}\right)=\mu\left(u_{R_{2}}\right)=m_{R}$), and the players play the jointly preferred outcome if they send the same message (i.e., $\xi\left(m_{L},m_{L}\right)=L$, $\xi\left(m_{R},m_{R}\right)=R$) and each obtain payoff $2$, and each player plays her risk-dominant action if they send different messages (i.e.,  $\xi\left(u_{L_{1}},m_{L},m_{R}\right)=\xi\left(u_{R_{2}},m_{R},m_{L}\right)=L$,  $\xi\left(u_{L_{2}},m_{L},m_{R}\right)=\xi\left(u_{R_{1}},m_{R},m_{L}\right)=R$). In the latter case (different messages) players with uniform preferences (i.e., $u_{L_1}$ and $u_{R_1}$) have an expected payoff of $2 \cdot\frac{8}{9}=\frac{16}{9}$, while those with non-uniform preferences (i.e., $u_{L_2}$ and $u_{R_2}$) have an expected payoff of $\frac{1}{9}$.

Observe that $\sigma$ is an equilibrium. In particular, a player with a non-uniform type (say $u_{L_1}$) would obtain only $\sfrac{1}{2}\cdot 1+\sfrac{1}{2}\cdot \sfrac{8}{9}=\sfrac{17}{18}$ by misreporting her preferred outcome, which is less than her equilibrium payoff of $\sfrac{1}{2}\cdot 2 + \sfrac{1}{2}\cdot \frac{1}{9}=\sfrac{19}{18}.$ The equilibrium payoff of a player with a uniform type is $\sfrac{1}{2}\cdot2+\sfrac{1}{2}\cdot\left(\sfrac{8}{9}\cdot2+\sfrac{1}{9}\cdot 0 \right)=\sfrac{17}{9}$.
Our next result shows that the $\sigma$, although being an equilibrium with miscoordination, is weakly communication proof. 

\begin{prop} \label{prop:counterexample} Equilibrium strategy $\sigma$ is weakly communication proof, and not strongly communication proof.
\end{prop}
\begin{proof}[Sketch of proof; see Appendix \ref{app:counterexproof} for the formal proof.] When the players send the same messages, they obtain their maximal feasible payoff, and thus there can be no Pareto-improving equilibrium. The key argument in the proof shows that after the players send different messages (i.e.,  $(m_L,m_R)$) there can be no Pareto-improving \emph{coordinated} equilibrium. In order for the coordinated equilibrium to be Pareto dominant, the uniform types ($u_{R_{1}}$ and $u_{L_{1}}$) must obtain an expected payoff of at-least $\sfrac{17}{9}$. Observe that in any coordinated equilibrium, types $L_1$ and $L_2$ must obtain the same payoff (as their payoff matrix differ only over outcomes with miscoordination), and similarly types $R_1$ and $R_2$ must obtain the same payoff. This implies that all types must obtain an expected payoff of at least $\sfrac{17}{9}>\sfrac{3}{2}$, which is impossible because the sum of payoffs of all outcomes is at most $3$. Next, the formal proof shows that there exist equilibria with miscoordination that CP trump $\sigma$ after $(m_L,m_R)$, but that these miscoordinated equilibria are CP trumped by other equilibria. This implies that $\sigma$ is weakly communication proof but not strongly communication proof.
\end{proof}
Weak communication-proofness holds in this setup due to the interplay between incomplete information and non-uniform coordination preferences. {Specifically,  we demonstrate} that non-uniform preferences  {give rise to} the existence of miscoordinated equilibria of games with incomplete information that are not dominated by coordinated equilibria. {This property persists despite the fact that}  all miscoordinated equilibria of the \emph{complete-information} games are dominated by coordinated equilibria.
\end{example}

{Suppose we change the distribution of types in Example \ref{example:counterexample} such that each of the four types have a probability of $\sfrac{1}{4}$.
In the following paragraphs we show that $\sigma$ is also an equilibrium of this game, but it is now CP-trumped by $\sigma_C$, and, thus, not weakly communication proof. This shows that once we go outside the class of uniform coordination preference types, the set of weakly communication proof equilibria (and probably also those of strongly communication proof equilibria) depends on the details of the distribution over types (while communication-proofness is independent of these details under the assumption of uniform coordination preferences, see Theorem \ref{Thm:multidimensional}).}

{To see that $\sigma$ is an equilibrium we need to check the players' incentives to follow $\sigma$ assuming their opponent chooses $\sigma$. Playing $L$ after $(m_L,m_L)$ and $R$ after $(m_R,m_R)$ is clearly optimal (if the opponent does the same). Consider the case of mixed messages $(m_L,m_R)$. A type $L_1$ would then prefers $L$ (which yields a payoff of $1$) over $R$ (a payoff of $\frac{1}{2}$), while a type $L_2$ prefers $R$ (a payoff of $\frac{1}{2}$) over $L$ (a payoff of $2-\frac{17}{2}$). The case for $R$ types is analogous. }
{Next, observe that a $L_1$ prefers message $m_L$ over $m_R$ type because the former yields a payoff of  $\sfrac{1}{2}\cdot2 + \sfrac{1}{2}\cdot2\cdot\sfrac{1}{2} = \sfrac{3}{2}$, while the latter yields a smaller payoff of either $\sfrac{1}{2}\cdot1 + \sfrac{1}{2}\cdot\sfrac{1}{2}$ (playing $R$) or $\sfrac{1}{2}\cdot1 + \sfrac{1}{2}\cdot2\cdot\sfrac{1}{2}$ (playing $L$). The incentives for $R$ types are analogous.}
{To see that $\sigma$ is CP-trumped by $\sigma_C$, note that the $\sigma$ equilibrium payoff after $(m_L,m_R)$ to $L_1$ and $R_1$ types is $2\cdot\sfrac{1}{2}=1$ and for $L_2$ and $R_2$ types it is $\sfrac{1}{2}$. Thus, strategy $\sigma_C$ (in which essentially a fair coin-toss coordinated action pair is chosen, which yields a payoff of \sfrac{3}{2} to all types) Pareto improves over $\sigma$ after $(m_L,m_R).$}

\section{Asymmetric Coordination Games\label{sec:Asymmetric-Coordination-Games}}
{Our baseline model assumes that both players' types have the same distribution $F$, and our solution concept focuses on symmetric equilibria. This
is done to simplify the notation and ease the exposition. In this section we adapt our model and results to asymmetric games.}

{\textbf{Adapted Model \,} Consider a setup similar to our baseline
model except that the distributions of the types of the two players'
positions differ: the type of player 1 is distributed according to
$F_{1}$ and the type of player 2 is distributed according to $F_{2}$.
As in the baseline model, both distributions are assumed to be atomless
with full support in $\left[0,1\right]$. Let $\langle\Gamma\left(F_{1},F_{2}\right),M\rangle$
denote the asymmetric coordination game with communication (to ease
notation, we assume that both players have the same set of messages
at their disposal). Let $\Sigma^{i}$ denote the set of all strategies
of player $i\in\left\{ 1,2\right\} $. We let $i$ denote the index
of one player and $j$ denote the index of the opponent.}

\begin{rem}
{In this context, the game $\langle\Gamma\left(F,F\right),M\rangle$ in which both players
have the same distribution of types now corresponds to a setup, in which
the payoff-irrelevant position of player 1 or player 2 is identifiable,
and the players can condition their play on their positions.} 
\end{rem}
{Given a strategy profile $\left(\sigma_{1},\sigma_{2}\right)$, let
$\pi_{u}^{i}\left(\sigma_{1},\sigma_{2}\right)$ denote the (interim)
payoff of type $u$ of player $i\in\left\{ 1,2\right\} $, and let
$\pi^{i}\left(\sigma_{1},\sigma_{2}\right)=\mathbb{E}_{u\sim F_{i}}\left[\pi_{u}^{i}\left(\sigma_{1},\sigma_{2}\right)\right]$
denote the ex-ante payoff of player $i\in\left\{ 1,2\right\} $. A
strategy profile $\left(\sigma_{1},\sigma_{2}\right)$ is an \emph{equilibrium} if $\pi_{u}^{1}\left(\sigma_{1},\sigma_{2}\right)\geq\pi_{u}^{1}\left(\sigma'_{1},\sigma_{2}\right)$
for each strategy $\sigma'_{1}\in\Sigma^{1}$ and for each type $u$
of player 1, and $\pi_{u}^{2}\left(\sigma_{1},\sigma_{2}\right)\geq\pi_{u}^{2}\left(\sigma_{1},\sigma'_{2}\right)$
for each strategy $\sigma'_{2}\in\Sigma^{2}$ and for each type $u$
of player $2$.}

\textbf{Adapted Key Properties \,} {We adapt the three key properties
of Section \ref{sec:Equilibrium-Strategies} as follows. Let $\mu_{u}^{i}\left(m_{i}\right)$
denote the probability, given message function $\mu^{i}$, that player
$i$ sends message $m_{i}$ if she is of type $u_{i}$. Let $\mu^{i}\left(m_{i}\right)=\mathbb{E}_{u\sim F_{i}}\left[\mu_{u}^{i}\left(m_{i}\right)\right]$
be the average (ex-ante) probability of player $i$ sending message
$m_{i}$. A strategy profile $\left(\sigma_{1},\sigma_{2}\right)$
is \emph{mutual-preference consistent} if whenever $u_{1},u_{2}<\sfrac{1}{2}$
then $\xi_{1}\left(m_{1},m_{2}\right)=\xi_{2}\left(m_{1},m_{2}\right)=L$
for all $m_{1}\in\mbox{supp}\left(\mu_{u}^{1}\right)$ and $m_{2}\in\mbox{supp}\left(\mu_{u}^{2}\right)$,
and whenever $u_{1},u_{2}>\sfrac{1}{2}$ then $\xi_{1}\left(m_{1},m_{2}\right)=\xi_{2}\left(m_{1},m_{2}\right)=R$
for all $m_{1}\in\mbox{supp}\left(\mu_{u}^{1}\right)$ and $m_{2}\in\mbox{supp}\left(\mu_{u}^{2}\right)$.}

{A strategy profile $\left(\sigma_{1},\sigma_{2}\right)$ is \emph{coordinated}
if $\xi_{1}\left(m_{1},m_{2}\right)=\xi_{2}\left(m_{1},m_{2}\right)\in\left\{ L,R\right\} $
for each pair of messages $m_{1}\in\mbox{supp}\left(\mu^{1}\right)$
and $m_{2}\in\mbox{supp}\left(\mu^{2}\right)$.}

{For any strategy profile $\sigma=\left((\mu^{1},\xi_{1}),(\mu^{2},\xi_{2})\right)\in\Sigma^{1}\times\Sigma^{2}$
and any message $m_{j}\in M$, define 
\[
\beta_{i}^{\sigma}(m_{j})=E_{u\sim F_{i}}\left[\sum_{m_{i}\in M}\mu_{u}^{i}(m_{i})\boldsymbol{1}_{\left\{ u\leq\xi_{i}(m_{i},m_{j})\right\} }\right]
\]
as the expected probability of player $i$ playing $L$ conditional
on player $j$ sending message $m_{j}\in M$. We say that strategy
profile $\sigma=\left(\sigma_{1},\sigma_{2}\right)$ has \emph{(essentially)
binary communication} if there are two pairs of numbers $0\le\underline{\beta}_{1}^{\sigma}\le\overline{\beta}_{1}^{\sigma}\le1$
and $0\le\underline{\beta}_{2}^{\sigma}\le\overline{\beta}_{2}^{\sigma}\le1$
such that for all messages $m\in M$ and each player $i\in\left\{ 1,2\right\} $
we have $\beta_{i}^{\sigma}(m)\in[\underline{\beta}_{i}^{\sigma},\overline{\beta}_{i}^{\sigma}]$;
for all messages $m\in M$ such that there is a type $u<\sfrac{1}{2}$
with $\mu_{u}^{j}(m)>0$ we have $\beta_{i}^{\sigma}(m)=\overline{\beta}_{i}^{\sigma}$;
and for all messages $m\in M$ such that there is a type $u>\sfrac{1}{2}$
with $\mu_{u}(m)>0$ we have $\beta_{i}^{\sigma}(m)=\underline{\beta}_{i}^{\sigma}$.}

Next we adapt the definition of left tendency. Consider a strategy profile $\sigma=\left(\sigma_{1},\sigma_{2}\right)$
that is coordinated and mutual-preference consistent and has binary
communication. Then there are $\alpha_{1}^{\sigma},\alpha_{2}^{\sigma}\in[0,1]$
such that, for each $i\in\left\{ 1,2\right\} $, 
\[
\underline{\beta}_{i}^{\sigma}=\left(1-F_{j}({\textstyle {\frac{1}{2}})}\right)\alpha_{i}^{\sigma}\mbox{ and }\overline{\beta}_{i}^{\sigma}=F_{j}({\textstyle {\frac{1}{2}})+\left(1-F_{j}({\textstyle {\frac{1}{2}})}\right)\alpha_{i}^{\sigma},}
\]
where $\alpha_{i}^{\sigma}$ is the probability of coordination on
$L$ conditional on player $i$ having type $u_{i}<\sfrac{1}{2}$
and player $j$ having type $u_{i}>\sfrac{1}{2}$. We refer to $\alpha^{\sigma}=\left(\alpha_{1}^{\sigma},\alpha_{2}^{\sigma}\right)$
as the \emph{left-tendency profile} of a strategy profile $\sigma$
that is coordinated and mutual-preference consistent and has binary
communication. It is simple to see that the set of strategies satisfying
the above three properties (coordination, mutual-preference consistency,
and binary communication) is essentially two-dimensional. The reason for this is that the
left-tendency profile $\alpha^{\sigma}=\left(\alpha_{1}^{\sigma},\alpha_{2}^{\sigma}\right)$
of such a strategy profile $\sigma$ describes all payoff-relevant
aspects. Two such strategy profiles $\sigma$ and $\sigma'$ with
the same left-tendency profile (i.e., with $\alpha^{\sigma}=\alpha^{\sigma'}$)
can only differ in the way in which the players implement the joint
lottery when they have different preferred outcomes. These implementation
differences are, however, not payoff-relevant, as the probability of the joint
lottery inducing the players to play $L$ remains the same.

The one-dimensional set of strategies that satisfy the three key properties in the baseline model is the subset of strategy profiles in the current setups for which the left-tendency profile is symmetric. In particular, profile $(\sigma_R,\sigma_R)$ (resp., $(\sigma_C,\sigma_C)$, $(\sigma_L,\sigma_L)$) induces left-tendency of $(0,0)$ (resp., $(0.5,0.5)$, $(1,1)$). Strategy profiles with different levels of left-tendency correspond to strategy profiles that treat the two players differently. In particular, the tendency profile $(1,0)$ (resp., $(0,1)$) correspond to a strategy profile that satisfies the above three properties in which whenever the players have different preferred outcomes they coordinate on playing the action preferred by player 1 (resp., player 2).

{\textbf{Adaptation of communication-proofness} Given a strategy
profile of the game $\langle\Gamma,M\rangle$ we denote the induced
``renegotiation'' game after a positive probability message pair
$m_{1},m_{2}\in M$ is sent by $\langle\Gamma(F_{m_{1}},F_{m_{2}}),\tilde{M}\rangle$.
For a strategy profile $\sigma'$ of such a renegotiation game $\langle\Gamma(G_{1},G_{2}),\tilde{M}\rangle$,
define the \emph{post-communication} expected payoffs for a player
$i$ of type $u$ by 
\[
\pi_{u}^{i}\left(\sigma'|G_{2}\right)=\mathbb{E}_{v\sim G_{2}}\left[\pi_{u,v}^{i}\left(\sigma'\right)\right]\equiv\int_{v=0}^{1}\pi_{u,v}^{i}\left(\sigma'\right)g_{2}\left(v\right)dv.
\]
Our definitions of weak and strong communication proofness remain the same in the setup with adding index $i$ to denote the player (i.e., $\pi_{u}^{i}(\cdot|G_{2})$ instead of $\pi_{u}(\cdot,\cdot|G_{2}),$ and with the adapted notation of $\sigma=(\sigma_1,\sigma_2)$ denoting a strategy profile.}

\textbf{Adapted Results \,} Our main result remains the same in the
setup of asymmetric coordination games. The proof, which is analogous
to the proof of Theorem \ref{thm:uniqueness}, is omitted for brevity. 
\begin{thm}
\label{Thm:n-players-1}Let $\sigma$ be a strategy
profile of $\langle\Gamma\left(F_{1},F_{2}\right),M\rangle$. The following statements are equivalent:
\begin{enumerate}
\item $\sigma$ is mutual-preference consistent, coordinated,
and has binary communication. 
\item $\sigma$ is a strongly communication-proof equilibrium
strategy.
\item $\sigma$ is a weakly communication-proof equilibrium
strategy. 
\end{enumerate}
\end{thm}
Proposition \ref{prop:Pareto-optimal}--\ref{prop-maximizing-coordinated}
and Corollary \ref{cor:strictly-improving-no-communication} can be adapted to the present setup analogously. It is straightforward to see that the asymmetric equilibrium
with the left-tendency profile $\left(1,0\right)$ (resp., $\left(0,1\right)$)
that always coordinates on the action preferred by Player 1 (resp.,
Player 2) is ex-ante Pareto efficient. This is in contrast to the symmetric
case, in which sometimes none of the symmetric communication-proof
equilibria are ex-ante Pareto efficient.

\section{Infrequent Non-Coordination Preferences\label{sec:Extreme-Types-with}}

{In this section, we show that an essentially unique strategy that satisfies the three key properties is strongly communication-proof in a setup in which a minority of the types have non-coordination preferences, that is, for which one of the actions is dominant. The distribution of these non-coordination types determines the unique level of left tendency that satisfies communication proofness.}

{Let $a<0$ and $b>1$. We extend the set of types to be the interval
$\left[a,b\right]$ (instead of being $[0,1]$ in the baseline model). Observe that action $L$ ($R$) is a dominant
action for any type $u<0$ ($u>1$) as coordinating on $R$ ($L$)
yields to such a type a negative payoff of $u<0$ ($1-u<0$). We call
types with a dominant action (i.e., $u<0$ or $u>1$) \emph{extreme},
and types that do not have a strictly dominant action (i.e., $u\in\left[0,1\right]$)
\emph{moderate}. We assume that the cumulative distribution of types
$F$ is continuous (atomless) and has full support in the interval
$\left[a,b\right]$.}

{We further assume that the extreme types are a minority both among
the agents who prefer action $R$ and among the agents who prefer
$L$, i.e., $F\left(0\right)<\sfrac{1}{2}F\left(\sfrac{1}{2}\right)\mbox{ and }1-F\left(1\right)<\sfrac{1}{2}\left(1-F\left(\sfrac{1}{2}\right)\right)$.
Next, we adapt the definitions of coordination and binary communication
to the current setup. The original definition of coordination is too
strong in the current setup, as, clearly, when extreme types with
different dominant actions meet they must miscoordinate. Thus, we
present a milder notion. A strategy is \emph{weakly coordinated} if
whenever two moderate types meet they never miscoordinate. Note that
the definition does not impose any restriction on what happens when
an extreme type meets a moderate type.}

The original definition of binariness is too weak in the current setup.
This is because coordinated strategies must allow for some miscoordination
between extreme types. This implies that an agent cares not only
about the average probability of the opponent playing left (i.e.,
$\beta^{\sigma}\left(m\right)$), but also about the total probability
of miscoordination. Thus, we strengthen binariness by requiring that
there exist two distributions of messages, which are used by all types
below $\sfrac{1}{2}$ and all types above $\sfrac{1}{2}$, respectively.
Formally, a strategy $\sigma=\left(\mu,\xi\right)$ has \emph{strongly
binary communication} if $\mu\left(u\right)=\mu\left(u'\right)$ if
either $u,u'\leq\sfrac{1}{2}$ or $u,u'>\sfrac{1}{2}$. It is easy
to see that the strategies $\sigma_{L},\sigma_{R},\sigma_{C}$ defined
in Section \ref{sec:Equilibrium-Strategies} all satisfy strongly
binary communication. Moreover, one can show, for any $\alpha\in\left[0,1\right]$,
that if there exists a strategy $\sigma$ that is coordinated, mutual-preference
consistent, and has binary communication with left tendency $\alpha$,
then there also exists strategy $\tilde{\sigma}$ with the same properties
that has strongly binary communication.

{Our next result shows that there exists, essentially, a unique communication-proof
equilibrium strategy that is coordinated, mutual-preference consistent,
and has strongly binary communication. 
\begin{prop}
\label{prop:uniquerpcoord} In a coordination game with communication
and with dominant action types, a strategy $\sigma$ that is coordinated,
mutual-preference consistent, and has strongly binary communication
is a {strongly} communication-proof equilibrium strategy
if and only if it has a left tendency of $\alpha=\frac{F\left(0\right)}{F\left(0\right)+\left(1-F\left(1\right)\right)}$.
\end{prop}}

{The formal proof is presented in Appendix }\ref{app:domtypes}. The key argument is that, for any equilibrium that satisfies the three key properties, the agent of type $u=\sfrac{1}{2}$ has to be indifferent between signalling $u\geq\sfrac{1}{2}$ and signalling $u\leq\sfrac{1}{2}$. For this to be true, we must have a strategy that counterbalances the differences between the frequency of extreme $L$-dominant types ($F(0)>0$) and the frequency  of extreme $R$-dominant types ($1-F(1)>0$). To see this, consider an adaptation of $\sigma_{L}$ to this setting by having extreme types follow their dominant actions  regardless of the sent messages. Note that $\sigma_{L}$ is no longer an equilibrium with extreme types.

Observe that a moderate type sending message
$m_{R}$ leads to coordination with probability one (sometimes on
$R$ and sometimes on $L$ depending on the opponent's message). A moderate type sending message $m_{L}$, in contrast, leads to coordination
(on $L$ only) with probability $F\left(1\right)<1$. This implies
that agents of type $u<\sfrac{1}{2}$ sufficiently close to $\sfrac{1}{2}$
strictly prefer sending message $m_{R}$ to sending message $m_{L}$
(as the former induces a higher probability of coordination). In turn, this implies that $\sigma_{L}$ is not an equilibrium. By contrast, the proof shows that for a strategy that satisfies the above three properties and that has a left tendency of $\alpha$, the probability of miscoordination is the same for all messages.

Appendix \ref{app:domtypes} also shows that a left tendency $\alpha$
communication-proof strategy that is coordinated and has strongly
binary communication can be implemented whenever $\alpha$ is a rational
number and the set of messages $M$ is sufficiently large. Irrational
$\alpha$-s can be approximately implemented by $\epsilon$-equilibria.

Observe that in the symmetric case ($F\left(0\right)=1-F\left(1\right)$),
strategy $\sigma_{C}$ is the essentially unique strongly communication-proof strategy with the above two properties. Further observe that in the
asymmetric case, the moderate types gain if the extreme types with
the same preferred outcome are more frequent than the extreme types
of the opposite preferred outcome. Specifically, suppose there are more
extreme ``leftists'' than extreme ``rightists'' (i.e., $F\left(0\right)>1-F\left(1\right)$). Then the essentially unique strongly communication-proof strategy with properties of coordination and strongly binary communication induces higher probability to coordinate on action $L$ (rather than on action $R$) whenever two moderate agents with different preferred outcomes meet.

\section{{Conclusion and} Discussion\label{sec:Related-Literature}}

{In this paper we adapt \citeauthor{blume1995communication}'s (\citeyear{blume1995communication}) notion of communication-proofness from sender-receiver games to games in which all players have incomplete information, all can communicate, and all can choose actions.\footnote{
Our notion is also related to notions of renegotiation-proofness that have been applied to repeated games (e.g., \citealp{farrell1989renegotiation,benoit1993renegotiation}), and to mechanisms and contracts in the presence of asymmetric information (e.g., \citealp{forges1994posterior,neeman2013ex,maestri2017dynamic,strulovici2017contract}).} We argue that this refinement seems appropriate both for sophisticated strategic agents who communicate until reaching a mutually beneficial solution, and for situations in which the agents' behavior is governed by a long-run evolutionary process.  We show that the refinement of communication-proofness selects a small subset of equilibria in two-action coordination games. In all of these equilibria the players never miscoordinate, they play their jointly preferred outcome whenever there is one, and they communicate only what is their preferred outcome without revealing how strongly they prefer this outcome over other coordinated outcomes. The behavior induced by communication-proof equilibria fits the stylized empirical facts, and has some desirable efficiency properties.}

Starting with the secret handshake argument provided in \citet{robson1990efficiency} (see also the earlier related notion of ``green beard effect'' in \citealp{hamilton1964genetical,dawkins1976selfish}), there is a sizable literature on the evolutionary analysis of costless pre-play communication before players engage in a complete information coordination game (as surveyed in Footnote \ref{footnote-evol-papers}). Suppose that a complete information coordination game has two Pareto-rankable equilibria. Then the Pareto-inferior equilibrium is not evolutionarily stable as it can be invaded by mutants who use a previously unused message as a secret handshake: if their opponent does not use the same handshake they simply play the Pareto-inferior equilibrium (as do all incumbents), but if their opponent also uses the secret handshake both sides play the Pareto-superior equilibrium. Our notion of communication-proofness extends the secret handshake argument to games with incomplete information by requiring that a communication-proof equilibrium should not to be Pareto-dominated by another equilibrium after any observed message profile.\footnote{Another closely related solution concept is \citeauthor{swinkels1992evolutionary}'s \citeyearpar{swinkels1992evolutionary} notion of robustness to equilibrium entrants. In a recent paper, \citet{newton2017shared} provides an evolutionary foundation for players developing the ability to renegotiate into a Pareto-better outcome (``collaboration'' in the terminology of \citeauthor{newton2017shared}).}

One argument that can be presented against the notion of communication-proofness is that non-communication-proof equilibria can be sustained by the following off-the-equilibrium path behavior: if any player proposes a joint deviation, then the equilibrium specifies that the opponent rejects the offer and that both players shift their behavior to playing an equilibrium that is bad for the proposer. This kind of off-the-equilibrium path proposer punishment would indeed deter players from suggesting joint deviations.\footnote{These kinds of proposer-punishing mechanisms are
explored in solution concepts of renegotiation-proofness that explicitly specify a structured renegotiation protocol, such as \citet{busch1995perfect}, \citet{santos2000alternating}, and \cite{Safronov-Strulovici}.}

Recall that we give the notion of communication-proofness two different interpretations: either we think of communication-proof equilibria as the plausible final outcomes of the deliberations of two rational and communicating agents, or we think of these equilibria as the stable outcomes of a long-run learning or evolutionary process.

Under each of these two interpretations one can counter the above proposer-punishing argument. Under the two rational deliberating agents interpretation, one can argue that agents may just have to be careful and subtle in the way they phrase their proposal. Suppose both agents face a situation (after initial messages are sent) in which they are about to play a Pareto-inferior action profile (relative to some possible available equilibrium in the induced game). They should then both realize that their proposer-punishing scheme, which prevents them from renegotiation, is not in their joint best interest and be able to overcome this.

Another related literature deals with stable equilibria in coordination games with private values, but without pre-play communication. \citet{sandholm2007evolution} (extending earlier results of \citealp{fudenberg1993learning,ellison2000learning}) shows that mixed Nash equilibria of the game with complete information can be purified in the sense of \citet{harsanyi1973games} in an evolutionarily stable way.\footnote{See also \citet{neary2017heterogeneity} who study coordination games without communication played on a graph, and provide sufficient conditions for heterogeneous equilibria with miscoordination to be stable.} Finally, two related papers analyze stag-hunt games with private values. \citet{baliga2004arms} show that introducing pre-play communication induces a new equilibrium in which the Pareto-dominant action profile is played with high probability. \citet{jelnov-et-al-2018} show that in some cases a small probability of another interaction can substantially affect the set of equilibrium outcomes in stag-hunt games with private values.


\appendix

\section{Formal Proofs\label{sec:Proofs}}

\subsection{Undominated Action Strategies\label{subsec:Any-Best-Reply-Strategy-is-cutoff}}

In what follows we show that our restriction to threshold action functions is without loss of generality, in the sense that each generalized strategy is dominated by a threshold strategy.

Let $\Gamma(F,G)$ be a coordination game without communication (possibly played after a pair of messages is observed in the original game $\langle\Gamma,M\rangle$). A generalized strategy is a measurable function $\eta:U\to\Delta\left(\{L,R\}\right)$ that describes a mixed action as a function of the player's type. A generalized strategy in $\Gamma(F,G)$ corresponds to a generalized action function $\xi:U\times M\times M\to\triangle\{L,R\}$, given a specific pair of observed messages $(m,m')$, i.e., $\eta\left(u\right)\equiv\xi\left(u,m,m'\right)$.

A pair of generalized strategies $\eta,\tilde{\eta}$ are almost surely realization equivalent (abbr., \emph{equivalent}), denoted by $\eta\approx\tilde{\eta}$, if they induce the same behavior with probability one, i.e., if
\[
\mathbb{E}_{u\sim F}\left[\eta\left(u\right)\neq\tilde{\eta}\left(u\right)\right]\equiv\int_{u\in U}f\left(u\right)\boldsymbol{1}_{\left\{ \eta\left(u\right)\neq\tilde{\eta}\left(u\right)\right\} }du=0.
\]
It is immediate that two equivalent generalized strategies always induce the same (ex-ante) payoff, i.e., that $\pi\left(\eta,\eta'\right)=\pi\left(\tilde{\eta},\eta'\right)$ for each generalized strategy $\eta'$.

A generalized strategy is a \emph{cutoff strategy} if there exists a type $x\in\left[0,1\right]$ such that $\eta(u)=L$ for each $u<x$ and $\eta(u)=R$ for each $u>x$. A generalized strategy $\eta$ is \emph{strictly dominated} by generalized strategy $\tilde{\eta}$ if $\pi\left(\eta,\eta'\right)<\pi\left(\tilde{\eta},\eta'\right)$ for any opponent's generalized strategy $\eta'$.

The following result shows that any generalized strategy is either equivalent to a cutoff strategy, or it is strictly dominated by a cutoff strategy.
\begin{lem}
\label{lem:dominated-by-cutoff-straegy} Let $\eta$ be a generalized strategy. Then there exists a cutoff strategy $\tilde{\eta}$, such that either $\eta$ is equivalent to $\tilde{\eta}$, or $\eta$ is strictly dominated by $\tilde{\eta}$.
\end{lem}
\begin{proof}
If $\mathbb{E}_{u\sim F}\left[\eta_{u}(L)\right]=1$ (resp., $\mathbb{E}_{u\sim F}\left[\eta_{u}(L)\right]=0$), then $\eta$ is equivalent to the cutoff strategy of always playing $L$ (resp., $R$). Thus, suppose that $\mathbb{E}_{u\sim F}\left[\eta_{u}(L)\right]\in\left(0,1\right)$. Let $x\in\left(0,1\right)$ be such that $F(x)=\mathbb{E}_{u\sim F}\left[\eta_{u}(L)\right]=\int_{u}\eta_{u}(L)f(u)du$. Let $\tilde{\eta}$ then be the cutoff strategy with cutoff $x$, i.e., $\tilde{\eta}_{u}(L)=1$ if $u\leq x$ and $\tilde{\eta}_{u}(L)=0$ if $u > x$. Assume that $\eta$ and $\tilde{\eta}$ are not equivalent, i.e., $\eta\not\approx\tilde{\eta}$. Let $\eta'$ be an arbitrary generalized strategy of the opponent. By construction, strategies $\eta$ and $\tilde{\eta}$ induce the same average probability of choosing $L$. Strategies $\tilde{\eta}$ and $\eta$ differ in that $\tilde{\eta}$ induces lower types to choose $L$ with higher probability, and higher types to choose $L$ with lower probability, i.e., $\eta_{u}(L)\leq\tilde{\eta}_{u}(L)$ for any type $u\leq x$ and $\eta_{u}(L)\geq\tilde{\eta}_{u}(L)$ for any type $u>x$. Since $\eta\not\approx\tilde{\eta}$ and $\mathbb{E}_{u\sim F}\left[\eta_{u}(L)\right]\in\left(0,1\right)$, it follows that the inequalities are strict for a positive measure of types, i.e.,
\[
0<\int_{u<x}f\left(u\right)\boldsymbol{1}_{\left\{ \eta\left(u\right)<\tilde{\eta}\left(u\right)\right\} }du\mbox{ and }0<\int_{u>x}f\left(u\right)\boldsymbol{1}_{\left\{ \eta\left(u\right)>\tilde{\eta}\left(u\right)\right\} }du.
\]
The fact that lower types 
always gain more (less) from choosing $L$ (\emph{R}) relative to higher types, with a strict inequality unless the opponent always plays $R$ ($L$), implies that $\pi\left(\eta,\eta'\right)<\pi\left(\tilde{\eta},\eta'\right)$.
\end{proof}

\subsection{Proof of Theorem \ref{thm:uniqueness}\label{subsec:uniqueness}}

We first prove the ``$1\Rightarrow2$'' part. Suppose that $\sigma=(\mu,\xi)\in\Sigma$ is mutual-preference consistent, coordinated, and has binary communication. As $\sigma$ is mutual-preference consistent it must satisfy $\mbox{supp}(F_{m})\subseteq[0,\sfrac{1}{2}]$ or $\mbox{supp}(F_{m})\subseteq[\sfrac{1}{2},1]$ for any message $m\in\mbox{supp}(\bar{\mu})$. Consider any $m,m'\in\mbox{supp}(\bar{\mu})$. There are three cases to consider. Suppose first that $\mbox{supp}(F_{m}),\mbox{supp}(F_{m'})\subseteq[0,\sfrac{1}{2}]$. Then as $\sigma$ is mutual-preference consistent we have that $\xi(m,m')=\xi(m',m)=L$. Thus $\xi$ describes best-reply behavior after this message pair. Moreover this behavior is the best possible outcome for any type in $[0,\sfrac{1}{2}]$ and thus for any type in $\mbox{supp}(F_{m})$ and $\mbox{supp}(F_{m'})$. The second case of $\mbox{supp}(F_{m}),\mbox{supp}(F_{m'})\subseteq[\sfrac{1}{2},1]$ is analogous.

Suppose, finally, that, w.l.o.g., $\mbox{supp}(F_{m})\subseteq[0,\sfrac{1}{2}]$ and $\mbox{supp}(F_{m'})\subseteq[\sfrac{1}{2},1]$. As $\sigma$ is \emph{coordinated} we have that $\xi(m,m')=\xi(m',m)=L$ or $\xi(m,m')=\xi(m',m)=R$. Action function $\xi$, therefore, again describes best-reply behavior. Moreover, one player always obtains her most preferred outcome. In order for a new strategy profile to improve the opponent's outcome, this new profile must require the former player to deviate from her most preferred outcome. Thus, no equilibrium $\sigma'$ in the game $\langle\Gamma(F_{m},F_{m'}),M\rangle$ Pareto dominates $\sigma$ after this message pair. This shows that action function $\xi$ is a best response to $\mu$ and to itself given $\mu$ and that, moreover, it cannot be {CP trumped}. It remains to show that the message function $\mu$ is optimal when the opponent chooses $\sigma=(\mu,\xi)$.

Consider type $u\in[0,\sfrac{1}{2}]$ and consider this type's choice of message. As $\sigma$ has binary communication and is coordinated, different messages $m\in M$ can only trigger different probabilities of coordinating on $L$ with a highest likelihood of such coordination for any message $m\in\mbox{supp}(\mu_{u})$. Therefore, type $u$ is indifferent between any message $m\in\mbox{supp}(\mu_{u})$ and weakly prefers sending any message $m\in\mbox{supp}(\mu_{u})$ to sending any message $m'\not\in\mbox{supp}(\mu_{u})$. An analogous statement holds for types $u\in[\sfrac{1}{2},1]$. This concludes the proof of the ``$1\Rightarrow2$'' part of
the theorem.

We prove the ``$3\Rightarrow1$'' part in three lemmas, one for each of the three properties.
\begin{lem}
\label{lemma:nomiscoordination} Every {weakly }communication-proof equilibrium strategy $\sigma=\left(\mu,\xi\right)$ is coordinated.
\end{lem}
\begin{proof}
We need to show that for any message pair $m,m'\in\mbox{supp}\left(\bar{\mu}\right)$,
\[
\mbox{ either }\xi(m,m')\geq\sup\left\{ u\mid\mu_{u}(m)>0\right\} \mbox{ or } \xi(m,m')\leq\inf\left\{ u\mid\mu_{u}(m)>0\right\} .
\]
Let $m,m'\in\mbox{supp}\left(\bar{\mu}\right)$ and assume to the contrary that
\[
\inf\left\{ u\mid\mu_{u}(m)>0\right\} <\xi(m,m')<\sup\left\{ u\mid\mu_{u}(m)>0\right\} .
\]
As $\sigma$ is an equilibrium, we have $\inf\left\{ u\mid\mu_{u}(m')>0\right\} <\xi(m',m)<\sup\left\{ u\mid\mu_{u}(m')>0\right\}$ because otherwise the sender of $m'$ would play $L$ with probability one or $R$ with probability one, in which case the best reply of the sender of message $m$ would be to play $L$ (or $R$) regardless of her type.

Let $x=\xi(m,m')$ and $x'=\xi(m',m)$. We now show that the equilibrium $(x,x')$ of the game without coordination $\Gamma\left(F_{m},F_{m'}\right)$ is CP-trumped by either $\sigma_{L}$, $\sigma_{R}$, or $\sigma_{C}$. There are three cases to be considered. Case 1: Suppose that $x,x'\le\sfrac{1}{2}$. We now show that in this case the equilibrium $(x,x')$ is CP-trumped by $\sigma_{R}$. Consider the player who sent message $m$.

Case 1a: Consider a type $u\le x$. Then we have
\[
(1-u)F_{m'}({\textstyle {\sfrac{1}{2}})+u\left(1-F_{m'}({\textstyle {\sfrac{1}{2}})}\right)\ge(1-u)F_{m'}(x'),}
\]
where the left-hand side is type a $u$ agent's payoff under strategy profile $\sigma_{R}$ and the right-hand side the payoff under strategy profile $(x,x')$. The inequality follows from the fact that $u\left(1-F_{m'}(\sfrac{1}{2})\right)\ge0$, and $F_{m'}(\sfrac{1}{2})\ge F_{m'}(x')$ follows from the fact that $F_{m'}$ is nondecreasing (as it is a cumulative distribution function). This inequality is strict for all $u$ except for $u=0$ in the case where $x'=\sfrac{1}{2}$.

Case 1b: Now consider a type $u$ with $x<u\le\sfrac{1}{2}$. Then we have
\[
(1-u)F_{m'}({\textstyle {\sfrac{1}{2}})+u\left(1-F_{m'}({\textstyle {\sfrac{1}{2}})}\right)>u\left(1-F_{m'}(x')\right),}
\]
where the left-hand side is a type $u$ agent's payoff under strategy profile $\sigma_{R}$ and the right-hand side is the payoff under strategy profile $(x,x')$. The inequality follows from the fact that by $u\le\sfrac{1}{2}$ we have that $1-u\ge u$, and therefore $(1-u)F_{m'}(\sfrac{1}{2})+u\left(1-F_{m'}(\sfrac{1}{2})\right)\ge u$.

Case 1c: Finally, consider a type $u>\sfrac{1}{2}$. Then we have $u>u\left(1-F_{m'}(x')\right),$ where the left-hand side is a type $u$ agent's payoff under strategy profile $\sigma_{R}$ and the right-hand side is the payoff under strategy profile $(x,x')$.

The analysis for the player who sent message $m'$ is analogous.

Case 2: Suppose that $x,x'\ge\sfrac{1}{2}$. The analysis is analogous to Case 1 if we replace $\sigma_{R}$ with $\sigma_{L}$.

Case 3: Suppose, w.l.o.g. for the remaining cases, that $x\le\sfrac{1}{2}\le x'$. The equilibrium $(x,x')$ in this case is Pareto-dominated by $\sigma_{C}$. To see this, consider the player who sent message $m$.

Case 3a: Consider a type $u\le x$. Then we have
\[
(1-u)\left[F_{m'}(\sfrac{1}{2})+\sfrac{1}{2}\left(1-F_{m'}(\sfrac{1}{2})\right)\right] + u\sfrac{1}{2}\left(1-F_{m'}(\sfrac{1}{2})\right)>(1-u)F_{m'}(x'),
\]
where the left-hand side is a type $u$ agent's payoff under strategy profile $\sigma_{C}$ and the right-hand side the payoff under strategy profile $(x,x')$. The inequality follows from the fact that we have $F_{m'}(x')=x\leq\sfrac{1}{2}$ due to $(x,x')$ being an equilibrium.

Case 3b: Now consider a type $u$ with $x<u\le\sfrac{1}{2}$. Then we have
\[
(1-u)\left[F_{m'}(\sfrac{1}{2}) + \sfrac{1}{2}\left(1-F_{m'}(\sfrac{1}{2}\right)\right] + u\sfrac{1}{2}\left(1-F_{m'}(\sfrac{1}{2})\right) > u\left(1-F_{m'}(x')\right),
\]
where the left-hand side is a type $u$ agent's payoff under strategy profile $\sigma_{C}$ and the right-hand side the payoff under strategy profile $(x,x')$. The inequality follows from the fact that by $u\le\sfrac{1}{2}$ we have $1-u\ge u$ and thus $(1-u)\left[F_{m'}(\sfrac{1}{2})+\sfrac{1}{2}\left(1-F_{m'}(\sfrac{1}{2})\right)\right] + u\sfrac{1}{2} \left(1-F_{m'}(\sfrac{1}{2})\right)\ge u$.

Case 3c: Finally, consider a type $u>\sfrac{1}{2}$. Then we have
\[
u\left[\left(1-F_{m'}(\sfrac{1}{2})\right) + \sfrac{1}{2} F_{m'}(\sfrac{1}{2}) \right] + (1-u) \sfrac{1}{2}F_{m'}(\sfrac{1}{2}) > u\left(1-F_{m'}(x')\right),
\]
where the left-hand side is a type $u$ agent's payoff under strategy
profile $\sigma_{C}$ and the right-hand side is the payoff under
strategy profile $(x,x')$. The inequality follows from the fact that
we have $F_{m'}(\sfrac{1}{2})>0$ and $F_{m'}(\sfrac{1}{2})\le F_{m'}(x')$.

The analysis for the player who sent message $m'$ is analogous.
\end{proof}
\begin{lem}
\label{lemma:binary} Every {weakly} communication-proof equilibrium strategy $\sigma$ has binary communication.
\end{lem}
\begin{proof}
Let $\sigma$ be a weakly communication-proof equilibrium strategy. Recall that
\[
\beta^{\sigma}(m)=\int_{u=0}^{1}\sum_{m'\in M}\mu_{u}(m')\boldsymbol{1}_{\{u\le\xi(m,m')\}}f(u)du.
\]
As $\sigma$ is coordinated by Lemma \ref{lemma:nomiscoordination}, the payoff to a type $u$ from sending $m\in \mbox{supp}(\overline{\mu})$ is
\[
(1-u)\beta^{\sigma}(m)+u\left(1-\beta^{\sigma}(m)\right).
\]
For a type $u<\sfrac{1}{2}$ the problem of choosing a message to maximize her payoffs is thus equivalent to choosing a message that maximizes $\beta^{\sigma}(m)$. We thus must have that there is a $\overline{\beta}^{\sigma}\in[0,1]$ such that for all $u<\sfrac{1}{2}$ and all $m\in\mbox{supp}(\mu_{u})$, we have $\beta^{\sigma}(m)=\overline{\beta}^{\sigma}$. Analogously, we must have a $\underline{\beta}^{\sigma}\in[0,1]$ such that for all $u>\sfrac{1}{2}$ and all $m\in\mbox{supp}(\mu_{u})$, we have $\alpha^{\sigma}(m)=\overline{\beta}^{\sigma}$. Clearly also $\underline{\beta}^{\sigma}\le\overline{\beta}^{\sigma}$. To extend the argument to unused messages $m\not\in \mbox{supp}(\overline{\mu})$ we rely on the full support assumption. Assume to the contrary that there is a message $m\not\in \mbox{supp}(\overline{\mu})$ with $\beta^{\sigma}(m)>\overline{\beta}^{\sigma}$ (resp., $\beta^{\sigma}(m)<\underline{\beta}^{\sigma}$). Then any sufficiently high (resp., low) type $u$ would strictly earn by deviating to sending message $m$ and playing $L$ (resp., $R$), which contradicts the supposition that $\sigma$ is an equilibrium strategy.
\end{proof}
\begin{lem}
\label{lemma:mpc} Every {weakly} communication-proof equilibrium strategy $\sigma$ is mutual-preference consistent.
\end{lem}
\begin{proof}
By Lemma \ref{lemma:nomiscoordination} a weakly communication-proof equilibrium strategy $\sigma=\left(\mu,\xi\right)$ is coordinated. Suppose that it is not mutual-preference consistent. Then there is either a pair $(m,m')$ such that there are types $u,v<\sfrac{1}{2}$ with $m\in\mbox{supp}(\mu_{u})$ and $m'\in\mbox{supp}(\mu_{v})$ such that play after $(m,m')$ is coordinated on $R$, or a pair $(m,m')$ such that there are types $u,v>\sfrac{1}{2}$ with $m\in\mbox{supp}(\mu_{u})$ and $m'\in\mbox{supp}(\mu_{v})$ such that play after $(m,m')$ is coordinated on $L$. In the former (resp., latter) case strategy $\sigma$ is CP-trumped by $\sigma_{R}$ (resp., $\sigma_{L}$) in the game $\left\langle \Gamma(F_{m},F_{m'}),\{m_{L},m_{R}\}\right\rangle $ because  $\sigma_{R}$ (resp., $\sigma_{L}$) does not affect the payoff of all types $u\geq\sfrac{1}{2}$ (resp., $u\leq\sfrac{1}{2}$), and it strictly improves the payoff to all types $u<\sfrac{1}{2}$ (resp., $u>\sfrac{1}{2}$).
\end{proof}

\subsection{Proofs of Section \ref{sec:Efficiency} (On Efficiency)}

\label{app:proofsefficiency}
\begin{proof}[Proof of Proposition \ref{prop:Pareto-optimal}]
By Theorem \ref{thm:uniqueness} and the discussion of the one-dimensional set of strategies satisfying the key properties in Section \ref{sec:Equilibrium-Strategies} a communication-proof equilibrium strategy $\sigma$'s  payoff is determined by its left tendency $\alpha\equiv\alpha^{\sigma}\in[0,1]$. This payoff is given by
\[
\pi_{u}(\sigma,\sigma)=(1-u)\left[F(\sfrac{1}{2})+\alpha\left(1-F(\sfrac{1}{2})\right)\right] + u(1-\alpha)\left[1-F(\sfrac{1}{2})\right],
\]
for each $u\in(0,\sfrac{1}{2}]$, and it is given by
\[
\pi_{u}(\sigma,\sigma)=(1-u)\alpha F(\sfrac{1}{2})+u\left[\left(1-F( \sfrac{1}{2})\right)+F( \sfrac{1}{2})(1-\alpha)\right].
\]
for each type $u\in(\sfrac{1}{2},1]$. The payoff to a type $u$ from  given social choice function $\phi$ is given by
\[
\pi_{u}\left(\phi\right) = \left(1-u\right) \mathbb{E}_{v}\phi_{u,v}\left(L,L\right) + u\mathbb{E}_{v}\phi_{u,v}\left(R,R\right).
\]

Now suppose that $\phi$ interim Pareto dominates $\sigma$. Then $\pi_{u}(\phi)\ge\pi_{u}(\sigma,\sigma)$ for all $u\in[0,1]$ with a strict inequality for a positive measure of $u$. As $\pi_{u}(\sigma,\sigma)$ is a convex combination of two payoffs, this implies that:
\begin{equation}
\mathbb{E}_{v}\phi_{u,v}\left(L,L\right)\ge F(\sfrac{1}{2}) + \alpha\left(1-F(\sfrac{1}{2})\right)\mbox{ for any } u \le \sfrac{1}{2},\mbox{ and} \label{eq-E_v}
\end{equation}
\begin{equation}
\mathbb{E}_{v}\phi_{u,v}\left(R,R\right) \ge \left(1-F(\sfrac{1}{2})\right)+F(\sfrac{1}{2})(1-\alpha)\mbox{ for any } u>\sfrac{1}{2},\label{eq-E_v2}
\end{equation}
with at least one of the inequalities holding strictly for a positive measure of types. Thus,
\[
\mathbb{E}_{v}\phi_{u,v}\left(L,L\right) = F(\sfrac{1}{2})\mathbb{E}_{\{v\le\sfrac{1}{2}\}} \phi_{u,v}\left(L,L\right) + \left(1-F(\sfrac{1}{2})\right) \mathbb{E}_{\{v>\sfrac{1}{2}\}}\phi_{u,v}\left(L,L\right),
\]
where, for instance, $\mathbb{E}_{\{v>\sfrac{1}{2}\}}$ denotes the expectation conditional on $v>\sfrac{1}{2}$. Substituting this last equality in Eq. (\ref{eq-E_v}) yields the following inequality
\[
F(\sfrac{1}{2})\mathbb{E}_{\{v\le\sfrac{1}{2}\}} \phi_{u,v}\left(L,L\right) + \left(1-F(\sfrac{1}{2})\right)\mathbb{E}_{\{v>\sfrac{1}{2}\}}\phi_{u,v}\left(L,L\right)\ge F( \sfrac{1}{2}) + \alpha\left(1-F(\sfrac{1}{2})\right)
\]
for any $u\le\sfrac{1}{2}$. The fact that $\mathbb{E}_{\{v\le\sfrac{1}{2}\}}\phi_{u,v}\left(L,L\right) \le 1$ implies that $\mathbb{E}_{\{v>\sfrac{1}{2}\}}\phi_{u,v}\left(L,L\right)\ge\alpha$ for any $u\le\sfrac{1}{2}$. An analogous argument (applied to Equation (\ref{eq-E_v2})) implies that $\mathbb{E}_{\{v<\sfrac{1}{2}\}}\phi_{u,v}\left(R,R\right)\ge1-\alpha$, for any $u>\sfrac{1}{2}$, with at least one of these inequalities holding strictly for a positive measure of types.
This implies that
\[
\mathbb{E}_{\{u<\sfrac{1}{2}\}}\mathbb{E}_{\{v>\sfrac{1}{2}\}}\phi_{u,v}\left(L,L\right)\ge\alpha\mbox{ and }\mathbb{E}_{\{u>\sfrac{1}{2}\}}\mathbb{E}_{\{v<\sfrac{1}{2}\}}\phi_{u,v}\left(R,R\right)\ge1-\alpha,
\]
with at least one of the two inequalities holding strictly. By the symmetry of $\phi$ we have $\phi_{u,v}(R,R)=\phi_{v,u}(R,R)$ and thus
\[
\mathbb{E}_{\{u<\sfrac{1}{2}\}}\mathbb{E}_{\{v>\sfrac{1}{2}\}}\phi_{u,v}\left(L,L\right)+\mathbb{E}_{\{u<\sfrac{1}{2}\}}\mathbb{E}_{\{v>\sfrac{1}{2}\}}\phi_{u,v}\left(R,R\right)>1,
\]
which contradicts $\phi_{u,v}$ being a social choice function.
\end{proof}
The proof of Proposition \ref{prop-maximizing-coordinated} uses the following lemma (which is of independent interest).

\begin{lem}
\label{lemma-maximizing-ex-ante-payoff} Let $\sigma\in\mathcal{E}$ be a coordinated equilibrium strategy. Then there is a communication-proof strategy $\sigma'$ such that either $\sigma$ and $\sigma'$ are interim payoff equivalent or $\sigma'$ interim Pareto dominates $\sigma$.
\end{lem}
\begin{proof}
Let $\sigma=(\mu,\xi)\in\mathcal{E}$ be coordinated. For each message $m\in M$, let $p_{m}\in\left[0,1\right]$ be the probability that the players coordinate on $L$, conditional on the agent sending message $m$:
\[
p_{m}=\sum_{m'\in M}\mu\left(\bar{m}'\right)\boldsymbol{1}_{\{\xi\left(m,m'\right)=L\}}.
\]
As $\sigma$ is coordinated, it follows that $1-p_{m}$ is the probability that the players coordinate on $R$, conditional on the agent sending message $m$.

Let $\bar{p}=\max_{m\in M}p_{m}$ be the maximal probability, and let $\underline{p}=\min_{m\in M}p_{m}$ be the minimal probability. By definition, $\underline{p}\le\bar{p}$. As $\sigma$ is an equilibrium strategy, $\underline{p}<\bar{p}$ implies that all types $u<\sfrac{1}{2}$ send a message inducing probability $\bar{p}$ and all types $u>\sfrac{1}{2}$ send a message inducing probability $\underline{p}$. Therefore, the expected payoff of a type $u\le\sfrac{1}{2}$ is given by $\pi_{u}\left(\sigma,\sigma\right)=\bar{p}\left(1-u\right)+\left(1-\bar{p}\right)u$, and the expected payoff of any type $u>\sfrac{1}{2}$ is equal to $\pi_{u}\left(\sigma,\sigma\right)=\underline{p}\left(1-u\right)+\left(1-\underline{p}\right)u.$ This is also true if $\underline{p}=\bar{p}$. Note that for types $u<\sfrac{1}{2}$, the expected payoff strictly increases in $\bar{p}$ and for types $u>\sfrac{1}{2}$ the type's expected payoff strictly decreases in $\underline{p}$.

We consider three cases. Suppose first that $\underline{p}\le\bar{p}\le F(\sfrac{1}{2})$. Then let $\sigma'=\sigma_{R}$. This strategy is also coordinated and its induced payoffs can be written in the same form as those for strategy $\sigma$ with $\underline{p}'=0$ and $\bar{p}'=F(\sfrac{1}{2})$. Thus, we get that $\pi_{u}\left(\sigma',\sigma'\right)\ge\pi_{u}\left(\sigma,\sigma\right)$ for every $u\in[0,1]$. This implies that $\sigma$ is either interim (pre-communication) payoff equivalent to or Pareto-dominated by $\sigma'=\sigma_{R}$.

The second case where $F(\sfrac{1}{2})\le\underline{p}\le\bar{p}$ is analogous to the first one, with $\sigma'=\sigma_{L}$. In the final case $\underline{p}<F(\sfrac{1}{2})<\bar{p}$. Let $\alpha\in[0,1]$ be such that $F(\sfrac{1}{2})+(1-F(\sfrac{1}{2}))\alpha=\bar{p}$ and let $\sigma'$ be a communication-proof strategy with left tendency $\alpha$. Then $\underline{p}\ge\alpha F(\sfrac{1}{2})$ and by construction $\sigma$ is either interim (pre-communication) payoff equivalent to or Pareto dominated by $\sigma'$.
\end{proof}
\begin{proof}[Proof of Proposition \ref{prop-maximizing-coordinated}]
By Lemma \ref{lemma-maximizing-ex-ante-payoff} we have that every coordinated equilibrium strategy $\sigma$ is interim (pre-communication) Pareto-dominated by some communication-proof strategy with some left tendency $\alpha\in[0,1]$ denoted by $\sigma_{\alpha}$. We thus have that $\pi\left(\sigma,\sigma\right)\leq\pi\left(\sigma_{\alpha},\sigma_{\alpha}\right)$.

The ex-ante expected payoff of to a $u$ type under strategy $\sigma_{\alpha}$ is given by
\[
\pi_{u}\left(\sigma_{\alpha},\sigma_{\alpha}\right)=\left(1-u\right)\left[F(\sfrac{1}{2}) + \alpha\left(1-F(\sfrac{1}{2})\right)\right] + u(1-\alpha)\left(1-F(\sfrac{1}{2})\right)\mbox{ for }u\le\sfrac{1}{2}\mbox{ and }
\]
\[
\pi_{u}\left(\sigma_{\alpha},\sigma_{\alpha}\right)=\left(1-u\right)\alpha F(\sfrac{1}{2}) + u\left[1-F(\sfrac{1}{2}) + (1-\alpha)F(\sfrac{1}{2})\right]\mbox{ for }u>\sfrac{1}{2}.
\]
It is straightforward to verify that $\pi_{u}\left(\sigma_{\alpha},\sigma_{\alpha}\right)=\alpha\pi_{u}\left(\sigma_{1},\sigma_{1}\right)+(1-\alpha)\pi_{u}\left(\sigma_{0},\sigma_{0}\right).$ for every $u$.

As $\sigma_{1}=\sigma_{L}$ and $\sigma_{0}=\sigma_{R}$ and as for all $u\in[0,1]$ $\pi_{u}\left(\sigma_{\alpha},\sigma_{\alpha}\right)$ is the same convex combination of $\pi_{u}\left(\sigma_{L},\sigma_{L}\right)$ and $\pi_{u}\left(\sigma_{R},\sigma_{R}\right)$, we have $\pi\left(\sigma_{\alpha},\sigma_{\alpha}\right)=\alpha\pi\left(\sigma_{1},\sigma_{1}\right)+(1-\alpha)\pi\left(\sigma_{0},\sigma_{0}\right)$, which implies that $\pi\left(\sigma,\sigma\right)\leq\pi\left(\sigma_{\alpha},\sigma_{\alpha}\right)\leq\max\left\{ \pi\left(\sigma_{L},\sigma_{L}\right),\pi\left(\sigma_{R},\sigma_{R}\right)\right\} $.
\end{proof}

\subsection{Proof of Theorem \ref{Thm:multidimensional}} \label{subsec:Proof-of-Theorem-multi}

The proof of Theorem \ref{Thm:multidimensional} mimics the proof of Theorem \ref{thm:uniqueness} except that Lemma \ref{lemma:nomiscoordination} has to be adapted somewhat as follows (this is the only place where one uses the assumption of {uniform} coordination preferences).
\begin{lem}
\label{lemma:nomiscoordinationmd}
Assume that the atomless distribution of types have {uniform} coordination preferences. Let $\sigma=\left(\mu,\xi\right)$ be a weakly communication-proof equilibrium strategy. Then it is coordinated.
\end{lem}
\begin{proof}
We need to show that for any message pair $m,m'\in\mbox{supp}\left(\bar{\mu}\right)$,
\[
\mbox{either }\xi(m,m')\geq\sup\left\{ \varphi_{u}\mid\mu_{u}(m)>0\right\} \mbox{ or }\xi(m,m')\leq\inf\left\{ \varphi_{u}\mid\mu_{u}(m)>0\right\}.
\]
Let $m,m'\in\mbox{supp}\left(\bar{\mu}\right)$ and assume to the contrary that
\[
\inf\left\{ \varphi_{u}\mid\mu_{u}(m)>0\right\} <\xi(m,m')<\sup\left\{ \varphi_{u}\mid\mu_{u}(m)>0\right\}.
\]
As $\sigma$ is an equilibrium, we must have $\inf\left\{ \varphi_{u}\mid\mu_{u}(m')>0\right\} <\xi(m',m)<\sup\left\{ \varphi_{u}\mid\mu_{u}(m')>0\right\}$. (Otherwise the $m'$ message sender would play $L$ with probability one or $R$ with probability one, in which case the $m$ message sender's best response would be to play $L$ (or $R$) regardless of her type). Let $x=\xi(m,m')$ and $x'=\xi(m',m)$. In what follows we will show that the equilibrium $(x,x')$ of the game without communication $\Gamma\left(F_{m},F_{m'}\right)$ is Pareto-dominated by either $\sigma_{L}$, $\sigma_{R}$, or $\sigma_{C}$ (all based on $\varphi_{u}$ instead of $u$).

There are three cases to be considered. Case 1: Suppose that $x,x'\le\sfrac{1}{2}$. We now show that in this case the equilibrium $(x,x')$ is Pareto-dominated by $\sigma_{R}$. Consider the player who sent message $m$.

Case 1a: Consider a type $u$ with $\varphi_{u}\le x$. Then we have
\[
u_{LL}F_{m'}(x')+\left(1-F_{m'}(x')\right)u_{LR}\le u_{LL}F_{m'}(\sfrac{1}{2}) + u_{LR}\left(1-F_{m'}(\sfrac{1}{2})\right) \le u_{LL}F_{m'}(\sfrac{1}{2}) + u_{RR}\left(1-F_{m'}(\sfrac{1}{2})\right),
\]
where the first expression is the type $u$ agent's payoff under strategy profile $(x,x')$ and the last expression is her payoff under strategy profile $\sigma_{R}$. The first inequality follows from $u_{LL}\ge u_{LR}$ and $F_{m'}(\sfrac{1}{2})\ge F_{m'}(x')$ by the fact that $F_{m'}$ is nondecreasing (as it is a CDF), and the second inequality follows from $u_{RR}\ge u_{LR}$. This inequality is strict when $u_{LL}>u_{LR}$ and $F_{m'}(\sfrac{1}{2})>F_{m'}(x')$ or when $u_{RR}>u_{LR}$.

Case 1b: Now consider a type $u$ with $x<\varphi_{u}\le\sfrac{1}{2}$. Then we have
\[
u_{RL}F_{m'}(x')+u_{RR}\left(1-F_{m'}(x')\right)\le u_{LL}F_{m'}(x')+u_{RR}\left(1-F_{m'}(x')\right)\le u_{LL}F_{m'}(\sfrac{1}{2}) + u_{RR}\left(1-F_{m'}(\sfrac{1}{2})\right),
\]
where the first expression is the type $u$ agent's payoff under strategy profile $(x,x')$ and the last expression is her payoff under strategy profile $\sigma_{R}$. The first inequality follows from $u_{LL}\ge u_{RL}$ and the second one from $F_{m'}(\sfrac{1}{2})\ge F_{m'}(x')$ and $u_{LL}\ge u_{RR}$. Note also that the second inequality follows from the assumption of {uniform} coordination preferences and $\varphi_{u}\le\sfrac{1}{2}$. This inequality is strict when $u_{LL}>u_{RL}$ or when $F_{m'}(\sfrac{1}{2})>F_{m'}(x')$ and $u_{LL}>u_{RR}$.

Case 1c: Finally, consider a type $u$ with $\varphi_{u}>\sfrac{1}{2}$. Then we have $u_{RR}>u_{RL}F_{m'}(x')+u_{RR}\left(1-F_{m'}(x')\right)$, where the right-hand side is the type $u$ agent's payoff under $(x,x')$ and the left-hand side is her payoff under $\sigma_{R}$. The inequality follows from the observation that $u_{RR}>u_{RL}$ because $u_{RR}>u_{LL}$ by the assumption of {uniform} coordination preferences, and $u_{LL}\ge u_{RL}$ by the fact that it is a coordination game.

The analysis for the player who sent message $m'$ is analogous.

Case 2: Suppose that $x,x'\ge\sfrac{1}{2}$. The analysis is analogous to Case 1 if we replace $\sigma_{R}$ with $\sigma_{L}$.

Case 3: Suppose, without loss of generality for the remaining cases, that $x\le\sfrac{1}{2}\le x'$. We show that the equilibrium $(x,x')$ in this case is Pareto-dominated by $\sigma_{C}$. Consider the player who sent message $m$.

Case 3a: Consider a type $u$ such that $\varphi_{u}\le x$. Then we have
\[
u_{LL}\left[F_{m'}(\sfrac{1}{2}) + \sfrac{1}{2}\left(1-F_{m'}(\sfrac{1}{2})\right)\right] + u_{RR}\sfrac{1}{2}\left(1-F_{m'}(\sfrac{1}{2})\right) > u_{LL}F_{m'}(x')+u_{LR}\left(1-F_{m'}(x')\right),
\]
where the right-hand side is the type $u$ agent's payoff under strategy profile $(x,x')$ and the left-hand side is her payoff under strategy profile $\sigma_{C}$. The inequality follows from the observation that $u_{RR}\ge u_{LR}$ and $F_{m'}(x')\le\sfrac{1}{2}$ by the fact that $F_{m'}(x')=x$ when $(x,x')$ is an equilibrium.

Case 3b: Now consider a type $u$ with $x<\varphi_{u}\le\sfrac{1}{2}$. Then we have
$$u_{RL}F_{m'}(x')+u_{RR}\left(1-F_{m'}(x')\right)\le u_{LL}F_{m'}(x')+u_{RR}\left(1-F_{m'}(x')\right)\le $$
$$u_{LL}\left[\sfrac{1}{2} + \sfrac{1}{2}F_{m'}(x')\right] + u_{RR}\sfrac{1}{2}\left(1-F_{m'}(\sfrac{1}{2})\right),$$
where the first expression is the type u agent's payoff under strategy profile $(x,x')$ and the last expression is her payoff under strategy profile $\sigma_{C}$. The first inequality follows from $u_{LL}\ge u_{RL}$ and the second one from $u_{LL}\ge u_{RR}$ by the assumption of {uniform} coordination preferences given $\varphi_{u}\le\sfrac{1}{2}$ and $F_{m'}(x')=x$ by $(x,x')$ being an equilibrium and $x<\sfrac{1}{2}$. The inequality is strict if $u_{LL}>u_{RL}$ or $u_{LL}>u_{RR}$.

Case 3c: Finally, consider a type $u$ with $\varphi_{u}>\sfrac{1}{2}$. Then we have
\begin{eqnarray*}
u_{RL}F_{m'}(x')+u_{RR}\left(1-F_{m'}(x')\right) & < & u_{RL} \sfrac{1}{2}F_{m'}(\sfrac{1}{2}) + u_{RR}\left[\left(1-F_{m'}(\sfrac{1}{2})\right) + \sfrac{1}{2}F_{m'} \sfrac{1}{2})\right]\\
 & \le & u_{LL}\sfrac{1}{2}F_{m'}(\sfrac{1}{2}) + u_{RR}\left[\left(1-F_{m'}(\sfrac{1}{2})\right) + \sfrac{1}{2}F_{m'}(\sfrac{1}{2})\right],
\end{eqnarray*}
where the first expression is a $u$ type's payoff under strategy profile $(x,x')$ and the last expression is her payoff under strategy profile $\sigma_{C}$. The first inequality follows from $u_{RR}>u_{LL}\ge u_{RL}$ by the assumption of {uniform} coordination preferences and from $\left(1-F_{m'}(\sfrac{1}{2})\right)\ge\left(1-F_{m'}(x')\right)$ as $F_{m'}$ is nondecreasing.

The analysis for the player who sent message $m'$ is analogous.
\end{proof}

\subsection{Proof of Proposition \ref{prop:counterexample}} \label{app:counterexproof}

We prove Proposition \ref{prop:counterexample} through a series of claims. Note that, given the equilibrium in question, after message pairs $(m_L,m_L)$ and $(m_R,m_R)$ no Pareto improvement is possible. It remains to be shown that, while there is a Pareto-improving equilibrium (with new communication) after message pair $(m_L,m_R)$, all Pareto-improving equilibria after message pair $(m_L,m_R)$ are themselves CP-trumped. The following claims refer to the situation after observed message pair $(m_L,m_R)$.

\begin{claim} \label{lem:alleq}
Suppose a further message pair leads to updated beliefs of $\alpha,1-\alpha$ of the L type being $L_{1}$ or $L_{2}$, respectively, and $\beta,1-\beta$ of the R type being $R_{1}$ or $R_{2}$, respectively. The following table provides the full list of Bayes Nash equilibria in the updated coordination game without (further) communication:
\[
\begin{array}{cc|ccc}
L_{1},L_{2} & R_{1},R_{2} & \mbox{payoffs } L_{1},L_{2};R_{1},R_{2} & \alpha & \beta \\ \hline
L,L & L,L & 2,2;1,1 & \in [0,1] & \in [0,1] \\
R,R & R,R & 1,1;2,2 & \in [0,1] & \in [0,1] \\
mix,R & mix,L & \sfrac23,\sfrac23;\sfrac23,\sfrac23  & \ge \sfrac23 & \ge \sfrac23 \\
mix,R & R,mix & \sfrac23,\sfrac23;\sfrac{16}{9},\sfrac19 & \ge \sfrac19 & \le \sfrac23 \\
L,mix & mix,L & \sfrac{16}{9},\sfrac19;\sfrac23,\sfrac23 & \le \sfrac23 & \ge \sfrac19 \\
L,mix & R,mix & \sfrac{16}{9},\sfrac19;\sfrac{16}{9},\sfrac19 & \le \sfrac19 & \le \sfrac19 \\
L,R & R,L & 2(1-\beta),\beta;2(1-\alpha),\alpha & \in [\sfrac19,\sfrac23] & \in [\sfrac19,\sfrac23]
\end{array}
\]
The last two columns provide the range of $\alpha$ and $\beta$ under which the various strategy profiles are equilibria.
\end{claim}
\begin{proof}
The proof follows straightforwardly from the observations that a probability of opponent playing action L (R) of $\sfrac13$ makes type $L_1$ ($R_1$) indifferent between actions L and R, while a probability of opponent playing action L (R) of $\sfrac89$ makes type $L_2$ ($R_2$) indifferent between actions L and R.
\end{proof}

\begin{claim} \label{lem:undomeq}
Of the equilibria provided in Claim \ref{lem:alleq} equilibria that are not given in the following table are CP trumped by other equilibria.

\[
\begin{array}{cc|cccc}
L_{1},L_{2} & R_{1},R_{2} & \mbox{payoffs } L_{1},L_{2};R_{1},R_{2} & \alpha & \beta & \alpha + \beta \\ \hline
L,L & L,L & 2,2;1,1 & \in [0,1] & \in [0,1] & \\
R,R & R,R & 1,1;2,2 & \in [0,1] & \in [0,1] & \\
L,mix & R,mix & \sfrac{16}{9},\sfrac19;\sfrac{16}{9},\sfrac19 & \le \sfrac19 & \le \sfrac19 & \\
L,R & R,L & 2(1-\beta),\beta;2(1-\alpha),\alpha & \in [\sfrac19,\sfrac{7}{18}] & \in [\sfrac19,\sfrac{7}{18}] & \le \sfrac12
\end{array}
\]
\end{claim}
\begin{proof}
All mixed equilibria except ((L,mix),(R,mix)) are Pareto dominated by either (L,L) or (R,R), see Claim \ref{lem:alleq}. Equilibrium ((L,R),(R,L)) (when $\alpha,\beta$ are outside the domain given in the table above) is dominated by a convex combination of (L,L) and (R,R) (it is dominated by (L,L) if $\alpha \geq \sfrac12$, dominated by (R,R) if $\beta \geq \sfrac12$, and dominated by a joint lottery that yields (L,L) with probability $1-2\beta$ and (R,R) with the remaining probability $2\beta$ if $\alpha,\beta< \sfrac12 <\alpha+\beta$).
\end{proof}

\begin{claim} \label{lem:L1beatsL2}
In any equilibrium of this game with or without additional communication, type $L_{1}$ $(R_{1})$ receives a payoff that is at least as high as that of type $L_{2}$ $(R_{2})$.
\end{claim}
\begin{proof}
The stage game payoff matrix for $L_{1}$ weakly exceeds that of $L_{2}$. Suppose there is an equilibrium in which $L_{2}$ expects a strictly higher payoff that $L_{1}$. Then $L_{1}$ can imitate $L_{2}$ and get at least the same payoff, a contradiction.
\end{proof}

\begin{claim} \label{lem:L1coord}
Consider an equilibrium of this game with communication that CP trumps the considered equilibrium ((L,R),(R,L)) after $(m_R,m_L)$ and that is not itself CP trumped by another strategy. If there is a message $m$ that only type $L_{1}$ $(R_{1})$ sends, then play after this message must be fully coordinated (against all opponent messages).
\end{claim}
\begin{proof}
A message $m$ that only $L_{1}$ sends reveals $L_{1}$ (leads to an updated belief that $m$ sender is of type $L_{1}$ with probability $\alpha=1$). Then only coordinated equilibria are possible as undominated equilibria - see Claim \ref{lem:undomeq} above for cases with $\alpha=1$. An analogous argument can be made for type $R_1$.
\end{proof}

\begin{claim} \label{lem:L1miscoord}
In any equilibrium of this game with communication that CP trumps the considered equilibrium ((L,R),(R,L)) and that is not itself CP trumped by another strategy, there can be no message sent with positive probability by type $L_{1}$ $(R_{1})$ that leads to coordinated play against all opponent messages.
\end{claim}
\begin{proof} Suppose there is a message $m$ that type $L_{1}$ sends with positive probability that leads to coordinated play (for all opponent messages). Then to Pareto-dominate the original equilibrium $L_{1}$ must expect a payoff of at least $\sfrac{16}{9}$.
Then type $L_{2}$ could imitate $L_{1}$ and also obtain the same payoff that $L_{1}$ obtains (because when play is coordinated both types receive the same payoff). Any other message $m'$ that type $L_{2}$ sends must also provide the same payoff. Suppose play after message $m'$ is not fully coordinated. Then type $L_{1}$ can send message $m'$ and imitate $L_{2}$'s behavior and receive a strictly higher payoff than $L_{2}$ does.
Thus, message $m'$ must also lead to fully coordinated play against all opponent messages.
Any message $m''$ sent by $L_{1}$ and not $L_{2}$ must lead to fully coordinated play as well by Claim \ref{lem:L1coord}. Thus, all messages sent with positive probability must lead to fully coordinated play. Given this, both types $L_{1}$ and $L_{2}$ receive a payoff greater or equal to $\sfrac{16}{9}$. But then $R_{1}$ can only obtain a payoff of at most $3-\sfrac{16}{9} = \sfrac{11}{9}$ (with $3$ being the maximal total payoff in any encounter), and the new equilibrium is no Pareto-improvement, a contradiction.
\end{proof}

\begin{claim}\label{lem:L1noreveal}
In any equilibrium of this game with communication that CP trumps the considered equilibrium ((L,R),(R,L)) and that is not itself CP trumped by another strategy, any message sent with positive probability by type $L_{1}$ $(R_{1})$ must also be sent by $L_{2}$ $(R_{2})$.
\end{claim}
\begin{proof} Suppose not and there is a message $m$ that reveals $L_{1}$. Then by Lemma \ref{lem:L1coord} this message must lead to coordinated play which contradicts Lemma \ref{lem:L1miscoord}.
\end{proof}

\begin{claim}\label{lem:onekindofmiscoord}
In any equilibrium of this game with communication that Pareto-dominates the considered equilibrium (L,R),(R,L) and that is not itself CP trumped by another strategy, there must be a message-pair $(m,m')$ sent with positive probability (by both types, respectively) that leads to an (L,R),(R,L) equilibrium with updated beliefs of $\alpha=P(L_{1}|m) \in [\sfrac19,\sfrac12]$ and $\beta=P(R_{1}|m) \in [\sfrac19,\sfrac12]$ that also satisfy $\alpha+\beta \le \sfrac12$.
\end{claim}
\begin{proof} By Claim \ref{lem:L1miscoord} every message sent must induce miscoordination against at least some opponent message. By Claim \ref{lem:L1noreveal} every message $L_{1}$ sends $L_{2}$ also sends. Thus, there must be a message $m$ that leads to an updated belief that $L_{1}$ sent this message of weakly more than $\sfrac19$ (analogously, $m'$ for R types). The only possible miscoordinated (and undominated) equilibrium given $(m,m')$ is given by (L,R),(R,L) - see Table above. We must have $\alpha=P(L_{1}|m) \in [\sfrac19,\sfrac12]$ and $\beta=P(R_{1}|m) \in [\sfrac19,\sfrac12]$ that also satisfy $\alpha+\beta \le \sfrac12$. Otherwise (L,R),(R,L) is Pareto dominated, a contradiction.
\end{proof}

\begin{claim}\label{lem:dominated}
Consider the stage game with $\mu=P(L_{1}|m) \in [\sfrac19,\sfrac12]$ and $\nu=P(R_{1}|m) \in [\sfrac19,\sfrac12]$ that also satisfy $\mu+\nu \le \sfrac12$. Then there is a strategy that CP trumps the equilibrium ((L,R),(R,L)).
\end{claim}
\begin{proof} Consider the following strategy. Types $L_{1}$ and $R_{1}$ send message $m_1$ with probability $1$. Types $L_{2}$ and $R_{2}$ send message $m_1$ with probability $\sfrac13$ and $m_2$ with probability $\sfrac23$. Continuation play is given by the following table:
\[
\begin{array}{c|cc}
 & m_1 & m_2 \\ \hline
m_1 & ((L,R),(R,L)) & ((L,L),(L,L)) \\
m_2 & ((R,R),(R,R)) & \sfrac12 ((L,L),(L,L)) + \sfrac12 ((R,R),(R,R))
\end{array}
\]
with strategies given for $((L_{1},L_{2})$,$(R_{1},R_{2}))$ in this sequence. Importantly ((L,R),(R,L)) is an equilibrium after $(m_1,m_1)$ because $\sfrac19 \le \alpha=P(L_{1}|m) = \sfrac{\nu}{\nu+(1-\nu)\sfrac13} \le \sfrac23$ (which is true for $\nu \le \sfrac25$) as $\nu \le \sfrac12-\sfrac19 = \sfrac{7}{18}<\sfrac25$ and analogously $\sfrac19 \beta \le \sfrac23$. In this equilibrium $L_{1}$ (resp., $R_{1}$) type have a payoff of $2(1-\nu)$ (resp., $2(1-\mu)$), which is the same payoff they get in the original (L,R),(R,L) equilibrium. Types $L_{2}$ and $R_{2}$ receive a payoff of more than $1$, which is more than they receive in the original ((L,R),(R,L)) equilibrium. Thus the new strategy CP trumps the original one.
\end{proof}

Claims \ref{lem:onekindofmiscoord}-\ref{lem:dominated} combine to prove that any strategy that CP trumps ((L,R),(R,L)) in the given game is itself CP trumped. This proves that ((L,R),(R,L)) is weakly communication proof. By Lemma \ref{lem:dominated} we also have that ((L,R),(R,L)) is not strongly communication proof. This proves Proposition \ref{prop:counterexample}.

Note, however, that any strategy that is coordinated and mutual-preference
consistent and has binary communication is strongly communication-proof
also in the general setting, and that only the ``$3\Rightarrow1$''
part of the main result fails without the assumption of {uniform}
coordination preferences. One can show that any communication-proof
equilibrium strategy must satisfy mutual-preference consistency, but,
possibly, need not satisfy the other two properties (namely, coordination
and binary communication).

\subsection{Proof of Proposition \ref{prop:uniquerpcoord}}

\label{app:domtypes} 
\begin{proof}[Proof of Proposition \ref{prop:uniquerpcoord}]
{Any strategy $\sigma$ that is coordinated, mutual-preference consistent,
and has strongly binary communication can be characterized by its
left tendency as follows.
Under such a mutual-preference consistent strategy players indicate
whether their type is below or above $\sfrac{1}{2}$. This means that
there are two disjoint sets of messages, $M_{L}$ and $M_{R}$, such
that players of type $u\le\sfrac{1}{2}$ send a message in $M_{L}$
and players of type $u>\sfrac{1}{2}$ send a message in $M_{R}$.
Also whenever two players both send messages in $M_{L}$ they then
play $L$ and if both send messages in $M_{R}$ they both play $R$.
The left tendency $\alpha=\alpha^{\sigma}$ then describes how moderate
players coordinate if one of them sends a message from $M_{L}$ and
the other sends a message from $M_{R}$. Specifically, $\alpha$
is the average probability that two such moderate players coordinate on $L$ (through
random message selection within the respective sets of messages),
while $1-\alpha$ is then the remaining probability that they coordinate
on $R$.}

{To prove the ``only if'' part, consider an arbitrary left tendency
of $\alpha\in[0,1]$. Then consider a player of type $\sfrac{1}{2}$
who needs to be indifferent between sending a message in $M_{L}$
and sending a message in $M_{R}$ for this strategy to be an equilibrium
strategy (such an indifference implies that any type below (resp., above) $\sfrac{1}{2}$ 
prefers to send a message in $M_{L}$ (resp., $M_{R}$)). If she sends a message in $M_{L}$ she coordinates on $L$
whenever either her opponent sends a message in $M_{L}$ (which happens
with probability $F(\sfrac{1}{2})$), or her moderate opponent sends
a message in $M_{R}$ (which happens with probability $F(1)-F(\sfrac{1}{2})$)
and the joint lottery yields the outcome $L$ (which happens with
probability $\alpha$). By contrast, she coordinates on $R$ whenever
her opponent sends a message in $M_{R}$ and the joint lottery yields
the outcome $R$ (which happens with probability $1-\alpha$). Therefore,
her expected payoff from sending a message in $M_{L}$ is given by
${\textstyle \frac{1}{2}F\left({\textstyle \frac{1}{2}}\right)+{\textstyle \frac{1}{2}\alpha\left(F(1)-F\left({\textstyle \frac{1}{2}}\right)\right)+{\textstyle \frac{1}{2}(1-\alpha)\left(1-F\left({\textstyle \frac{1}{2}}\right)\right).}}}$
Similarly, her expected payoff from sending a message in $M_{R}$
is ${\textstyle \frac{1}{2}\left(1-F\left({\textstyle \frac{1}{2}}\right)\right)}+{\textstyle \frac{1}{2}\alpha F\left({\textstyle \frac{1}{2}}\right)+{\textstyle \frac{1}{2}(1-\alpha)\left(F\left({\textstyle \frac{1}{2}}\right)-F\left(0\right)\right).}}$
Simple calculation shows that her expected payoff from sending a message in $M_{L}$ is
equal to her expected payoff from sending a message in $M_{R}$ iff
$\alpha=F(0)/\left(F(0)+(1-F(1))\right)$, as required.}

{To prove the ``if'' direction, we need to show that a strategy that satisfies the three adapted key properties and has a left tendency of $\alpha=F(0)/\left(F(0)+(1-F(1))\right)$
is both an equilibrium and a communication-proof strategy. To prove
the latter condition the same arguments as in the relevant parts of
the proof of the ``if'' direction of Theorem \ref{thm:uniqueness}
apply directly. It remains to show that such a strategy is an equilibrium
strategy. We have already shown that the message function is a best
reply to itself and the action function. All that remains to prove
is that the action function is a best reply to the given strategy.
It is easy to see that playing $L$ is the optimal strategy when both
players send a message in $M_{L}$ and thus are of type $u<\sfrac{1}{2}$.
In doing so, they coordinate on their most preferred outcome with
probability one. Similarly, playing $R$ after two messages in $M_{R}$
is clearly optimal.} 

{Now suppose that one player sends a message in
$M_{L}$ and the other player sends a message in $M_{R}$. There are
two possibilities. Either they are now supposed to both play $L$
(unless they are an extreme $R$ type) or they are now supposed to
both play $R$ (unless they are an extreme $L$ type). Consider first
the person who sends a message in $M_{L}$ and therefore be of type
$u<\sfrac{1}{2}$. Suppose that the two players are expected to coordinate
on $R$. Since her opponent sent a message in $M_{R}$, our $M_{L}$
sender expects $R$ with a probability of one (as all $M_{R}$ senders
are of type $u>\sfrac{1}{2}$, which excludes $L$-dominant action
types). But then our $M_{L}$ sender of any type $u>0$ has a strict
incentive to play $R$ as well. Now suppose that the two players are
expected to coordinate on $L$. Then our $M_{L}$ sender expects her
opponent to play $L$ with a probability of $\left(F(1)-F(\sfrac{1}{2})\right)/\left(1-F(\sfrac{1}{2})\right)$,
which is the conditional probability of an $R$-type to be moderate,
which by assumption is greater than or equal to $\sfrac{1}{2}$. Playing
$L$ in this case is therefore optimal for all $M_{L}$ senders. That
the $M_{R}$ sender has the correct incentives in her choice of action
after any mixed-message pair (one in $M_{R}$ and one in $M_{L}$)
is proven analogously and requires the assumption that $F(0)/F(\sfrac{1}{2})\le\sfrac{1}{2}$.} 
\end{proof}
{We now show when one can implement a coordinated, mutual-preference
consistent strategy with binary communication with the required left
tendency of $\alpha=F(0)/\left(F(0)+(1-F(1))\right)$. This implementation
requires two things. First, $\alpha$ needs to be a rational number,
and second, the message space needs to be sufficiently large.}\footnote{{The method for implementing a binary joint lottery of $\alpha$ and
$1-\alpha$ is based on} \citet{maschler1968repeated} {and relies on
$\alpha$ being a rational number. In order to deal with irrational
$\alpha$-s one needs either to slightly weaken the result to show
that there exists a communication-proof $\epsilon$-equilibrium strategy
(in which each type of each player gains at most $\epsilon$ from
deviating) for any $\epsilon>0$, or to allow an infinite set of messages
or a continuous ``sunspot.''}} {Let $\alpha=k/n$, and assume that $\left|M\right|\geq2n$.
Denote $2n$ distinct messages as $\left\{ m_{L,1},\ldots,m_{L,n},m_{R,1},\ldots,m_{R,n}\right\} \in M$,
where we interpret sending message $m_{L,i}$ (resp., $m_{R,i}$) as expressing a preference
for $L$ (resp., $R$) and choosing at random the number $i$ from the set of numbers
$\left\{ 1,\ldots,n\right\} $ in the joint lottery described below.
We arbitrarily interpret any message $m\in M\setminus\left\{ m_{L,1},\ldots,m_{L,n},m_{R,1},\ldots,m_{R,n}\right\} $
as equivalent to $m_{L,1}$. Given message $m\in M,$ let $i\left(m\right)$
denote its associated random number, e.g., $i\left(m_{L,j}\right)=j$.
Let $M_{R}=\left\{ m_{R,1},\ldots,m_{R,n}\right\} $ and $M_{L}=M\setminus M_{R}$.
Then $\sigma_{\alpha}=\left(\mu_{\alpha},\xi_{\alpha}\right)$ can
be defined as follows:} 
\[
\mu_{\alpha}\left(u\right)=\begin{cases}
\frac{1}{n}m_{L,1}+\ldots+\frac{1}{n}m_{L,n} & u\leq\frac{1}{2}\\
\frac{1}{n}m_{R,1}+\ldots+\frac{1}{n}m_{R,n} & u>\frac{1}{2},
\end{cases}\,\,\,\,\xi_{\alpha}\left(m,m'\right)=\begin{cases}
0 & \left(m,m'\right)\in M_{R}\times M_{R}\\
0 & \left(m,m'\right)\not\in M_{L}\times M_{L}\\
 & \mbox{ and }(i\left(m\right)+i\left(m'\right)\mod n)\,>k\\
1 & \mbox{otherwise}.
\end{cases}
\]

{Thus, $\mu_{\alpha}$ induces each agent to reveal whether her preferred
outcome is $L$ or $R$, and to uniformly choose a number between
$1$ and $n$. If both agents share the same
preferred outcome they play it. Otherwise, moderate types coordinate
on $L$ if the sum of their random numbers modulo $n$ is at most
$k$, and coordinate on $R$ otherwise. Extreme types play their strictly
dominant action.}
\section{More on Properties of Strategies} \label{app:properties}

In this appendix we demonstrate that no single one of the three properties (mutual-preference consistency, coordination, and binary communication) is implied by the other two. Clearly a strategy that has binary communication and is coordinated must be an equilibrium. No other combination of two of the three properties implies that a strategy is an equilibrium. Finally, we also define what it means for a strategy to be ordinal preference-revealing and show that this is implied by it being mutual-preference consistent.

Consider the following strategy $\sigma=(\mu,\xi)$ in the game with communication with a message set $M$ that contains at least three elements. Let $m_{L}^{1},m_{L}^{2},m_{R}\in M$, let

\[
\mu(u)=\left\{ \begin{array}{cc}
m_{L}^{1} & \mbox{ if }u\le\sfrac{1}{4}\\
m_{L}^{2} & \mbox{ if }\sfrac{1}{4}<u\le\sfrac{1}{2}\\
m_{R} & \mbox{ if }u>\sfrac{1}{2}
\end{array}\right.,
\]
and let $\xi$ be such that $\xi(m_{L}^{i},m_{L}^{j})=L$ for all $i,j\in\{1,2\}$, $\xi(m_{R},m_{R})=R$, $\xi(m_{L}^{1},m_{R})=\xi(m_{R},m_{L}^{1})=R$, and $\xi(m_{L}^{2},m_{R})=\xi(m_{R},m_{L}^{2})=L$. This strategy is mutual-preference consistent and coordinated but does not have binary communication. It is not an equilibrium as types $u\le\sfrac{1}{4}$ would strictly prefer to send message $m_{L}^{2}$.

Consider the following strategy $\sigma=(\mu,\xi)$ in the game with communication with a message set $M$ that contains at least two elements. Let $m_{L},m_{R}\in M$, let $\mu(u)=m_L$ if $u\le\sfrac{1}{2}$ and $\mu(u)=m_R$ if $u>\sfrac{1}{2}$.
Let $\xi$ be such that $\xi(m_{L},m_{L})=L$, $\xi(m_{R},m_{R})=R$, $\xi(m_{L},m_{R})=\sfrac{1}{4}$, and $\xi(m_{R},m_{L})=\sfrac{3}{4}$. This strategy is mutual-preference consistent, has binary communication, but is not coordinated. For almost all type distributions $F$ this is not an equilibrium: it is only an equilibrium if $F$ satisfies $$\left(F(\sfrac{3}{4})-F(\sfrac{1}{2})\right)/\left(1-F(\sfrac{1}{2})\right)=\sfrac{1}{4}
\mbox{ and } F(\sfrac{1}{4})/F(\sfrac{1}{2})=\sfrac{3}{4}.$$
Finally, for a coordinated strategy with binary communication that is not mutual-preference consistent, consider the equilibrium that always leads to coordination on $L$ for any pair of messages.

Note also that an equilibrium does not necessarily satisfy any of the three properties. The interior cutoff babbling equilibria mentioned in Section \ref{sec:Equilibrium-Strategies} are not coordinated and not mutual-preference consistent. The equilibrium of Example \ref{exa:high-ex-ante-miscoordination} does not have binary communication.

Call a strategy $\sigma=(\mu,\xi)\in\Sigma$ \emph{ordinal preference-revealing} if there exist two nonempty, disjoint, and exhaustive subsets of $\mbox{supp}(\bar{\mu})$ denoted by $M_{L}$ and $M_{R}$ (i.e., $\mbox{supp}(\bar{\mu})=\dot{M_{L}\bigcup M_{R}}$) such that if $u<\sfrac{1}{2}$, then $\mu_{u}(m)=0$ for each $m\in M_{R}$, and if $u>\sfrac{1}{2}$, then $\mu_{u}(m)=0$ for each $m\in M_{L}$. With an ordinal preference-revealing strategy a player indicates her ordinal preferences. A strategy $\sigma$ that is mutual-preference consistent is also ordinal preference-revealing (but not vice versa). Suppose not. Then there is a message $m$ and two types $u<\sfrac{1}{2}$ and $v>\sfrac{1}{2}$ such that $\mu_{u}(m),\mu_{v}(m)>0$. But then no matter how we specify $\xi(m,m)$ we get either that if two types $u$ meet they do not coordinate on $L$ with probability one or if two types $v$ meet they do not coordinate on $R$ with probability one.

\section{Multiple Rounds of Communication\label{subsec:Multiple-Rounds-of}}

Consider a variant of the coordination game with communication in which players have $T\ge1$ of rounds of communication. In each such round players simultaneously send messages from the set $M$. Players observe messages after each round and can, thus, condition their message choice and then their final action choice on the history of observed message pairs up to the point in time where they take their message or action decision. Renegotiation then possibly takes place once at the end of this communication phase but before the final action choices are made. Let $\mathcal{M}=\bigcup_{t=0}^{T-1}\left(M\times M\right)^{\text{t}}$, where $\left(M\times M\right)^{0}=\emptyset$.

A (pure) \emph{message protocol} is a function $\mathfrak{m}:\mathcal{M}\rightarrow M$ that describes the message sent by an agent as a deterministic function of the message profiles observed in the previous rounds of communication. Let $\mathfrak{M}$ be the set of all message protocols. A strategy $\sigma=(\mu,\xi)$ is a pair where $\mu:U\to\Delta(\mathfrak{M})$ denotes the \emph{message function}, prescribing a (possibly random) message protocol for each type, and $\xi:\left(M\times M\right)^{T}\to U$ denotes the \emph{action function} by means of describing the cutoff (the highest possible value of $u$) for the two players to choose action $L$ after observing the final message history. Renegotiation is modeled, as in the main text, as a possibility for the two players to play an equilibrium of a new game with another round of communication after all messages are sent, possibly using a different message set.

Next, we adapt the notion of binary communication to fit multiple rounds of communication. For any message protocol $\mathfrak{m}\in\mathfrak{M}$, let $\beta^{\sigma}(\mathfrak{m})$ denote the expected probability of a player's opponent playing $L$ conditional on the player following message protocol $\mathfrak{m}\in\mathfrak{M}$ and the opponent following strategy $\sigma=(\mu,\xi)\in\Sigma$. We say that strategy $\sigma$ has \emph{binary communication} if there are two numbers $0\le\underline{\beta}^{\sigma}\le\overline{\beta}^{\sigma}\le1$ such that for all message protocols $\mathfrak{m}\in\mathfrak{M}$ we have $\beta^{\sigma}(\mathfrak{m})\in[\underline{\beta}^{\sigma},\overline{\beta}^{\sigma}]$, for all message protocols $\mathfrak{m}\in\mathfrak{M}$ such that there is a type $u<\sfrac{1}{2}$ with $\mu_{u}(\mathfrak{m})>0$ we have $\beta^{\sigma}(\mathfrak{m})=\overline{\beta}^{\sigma}$, and for all message protocols $\mathfrak{m}\in\mathfrak{M}$ such that there is a type $u>\sfrac{1}{2}$ with $\mu_{u}(\mathfrak{m})>0$ we have $\beta^{\sigma}(\mathfrak{m})=\underline{\beta}^{\sigma}$. That is, binary communication implies that players use just two kinds of message protocols: any message protocol used by types $u<\sfrac{1}{2}$ induces the consequence of maximizing the probability of the opponent to play $L$, and any message protocol used by types $u>\sfrac{1}{2}$ induces the opposite consequence of maximizing the probability of the opponent to play $R$.

Theorem \ref{thm:uniqueness}, together with Propositions \ref{prop:Pareto-optimal}-\ref{prop-maximizing-coordinated}, holds in this setting with minor adaptations to the proof (omitted for brevity). Thus, regardless of the length of the pre-play communication, agents can reveal only their preferred outcome (but not the strength of their preference), and, regardless of having access to additional rounds of communication, they cannot improve the ex-ante expected payoff relative to the payoff induced by a single round of communication with a binary message.

\bibliographystyle{chicago}
\bibliography{CBC-bibil}

\end{document}